\newcommand{\paperTitle}{Analytic relationship of relative synchronizability to network structure and motifs}
\newcommand{\theKeywords}{synchronization, complex networks, motifs, feedback loops, feedforward loops}

\newcommand{\abstractText}{Synchronization phenomena on networks have attracted much attention in studies of neural, social, economic, and biological systems, yet we still lack a systematic understanding of how relative synchronizability relates to underlying network structure.
Indeed, this question is of central importance to the key theme of how dynamics on networks relate to their structure more generally.
We present an analytic technique to directly measure the relative synchronizability of noise-driven time-series processes on networks, in terms of the directed network structure.
We consider both discrete-time auto-regressive processes and continuous-time Ornstein-Uhlenbeck dynamics on networks.
Our technique builds on computation of the network covariance matrix in the space orthogonal to the synchronized state, enabling it to be more general than previous work in not requiring either symmetric (undirected) or diagonalizable connectivity matrices, and allowing arbitrary self-link weights.
More importantly, our approach quantifies the relative synchronisation specifically in terms of the contribution of process motif (walk) structures. We demonstrate that in general the relative abundance of process motifs with convergent directed walks (including feedback and feedforward loops) hinders synchronizability.
We also reveal subtle differences between the motifs involved for discrete or continuous-time dynamics.
Our insights analytically explain several known general results regarding synchronizability of networks, including that small-world and regular networks are less synchronizable than random networks.}

\documentclass[
 aps, %
 amsmath,amssymb,
 pre,superscriptaddress,showkeys,showpacs, %
 reprint,%
 numerical,%
 nofootinbib,
 nobalancelastpage, %
]{revtex4-2}

\newcommand{\fullSubSecRef}[2]{\secRef{#2}} %
\newcommand{\app}[1]{Supplementary Information \ref{app:#1}}
\newcommand{\fig}[1]{Fig.~\ref{fig:#1}}
\newcommand{\subfigs}[2]{Fig.~\ref{fig:#1}-\subref{fig:#2}}
\newcommand{\subrefs}[2]{\subref{fig:#1}-\subref{fig:#2}}
\newcommand{\eq}[1]{Eq.~(\ref{eq:#1})}
\newcommand{\eqs}[2]{Eq.~(\ref{eq:#1}-\ref{eq:#2})}

\newcommand{\secRef}[1]{Section \ref{sec:#1}}
\newcommand{\fn}[1]{footnote \ref{fn:#1}}

\newcommand{\theoremRef}[1]{Theorem \ref{th:#1}}
\newcommand{\lemmaRef}[1]{Lemma \ref{lem:#1}}

\newcommand{\generalMotifClosedHeight}{0.068\textwidth}
\newcommand{\generalMotifOpenHeight}{0.083\textwidth}
\newcommand{\motifFigWidth}{0.10\textwidth}
\newcommand{\motifFigHeight}{0.08\textwidth}

\newcommand{\walkmotif}[3]{\mathbf{w}_{#1 \rightarrow #2,#3}}
\newcommand{\walkprod}[6]{\mathbf{w}_{#1 \rightarrow #2,#3}^{#4 \rightarrow #5,#6}}

\usepackage{amsthm}
\usepackage{graphicx}%
\usepackage{subfigure} %
\usepackage{supertabular,array}
\usepackage{hyphenat} %
\usepackage{multirow}
\usepackage[sort&compress]{natbib}

\usepackage{xcolor}

\newtheorem{theorem}{Theorem}
\newtheorem{lemma}{Lemma}

\ifx\pdfoutput\undefined
\else
    \usepackage[bookmarks=false,colorlinks=false,citecolor=black,linkcolor=black,breaklinks=true,pdftex]{hyperref}
  	\hypersetup{
      	pdfauthor = {Joseph T. Lizier, Frank Bauer, Fatihcan M. Atay and J\"urgen Jost},
        pdftitle = {\paperTitle},
  	    pdfsubject = {Version \today},
      	pdfkeywords = {\theKeywords},
        pdfcreator = {LaTeX with hyperref package},
  	    pdfproducer = {pdftex},
  	    pdfstartview = {FitH}}
\fi

\begin{document}

\title{\paperTitle}

\author{Joseph T. Lizier}
\email[]{joseph.lizier@sydney.edu.au}
\affiliation{School of Computer Science and Centre for Complex Systems, Faculty of Engineering, The University of Sydney, NSW 2006, Australia}
\affiliation{Max Planck Institute for Mathematics in the Sciences,
Inselstra{\ss}e 22, 04103 Leipzig, Germany}

\author{Frank Bauer}
\affiliation{Max Planck Institute for Mathematics in the Sciences,
Inselstra{\ss}e 22, 04103 Leipzig, Germany}
\affiliation{Department of Mathematics, Harvard University, 1 Oxford Street, Cambridge MA 02138, USA}

\author{Fatihcan M. Atay}
\affiliation{Department of Mathematics, Bilkent University, 06800 Ankara, Turkey}
\affiliation{Max Planck Institute for Mathematics in the Sciences,
Inselstra{\ss}e 22, 04103 Leipzig, Germany}

\author{J\"urgen Jost}
\affiliation{Max Planck Institute for Mathematics in the Sciences,
Inselstra{\ss}e 22, 04103 Leipzig, Germany}
\affiliation{Santa Fe Institute, 1399 Hyde Park Road, Santa Fe, NM 87501, USA}

\date{May 17, 2023}

\begin{abstract}
\abstractText
\end{abstract}


\maketitle

\section{Introduction}
\label{sec:intro}

The manner in which systems of coupled oscillators achieve synchronization is one of the most widely-occuring, compelling, and oldest examples analysed in complex systems science \cite{aren08,Pikovsky01,stro03}.
Synchronization has been observed and studied in many different and seemingly unrelated fields, including in swarms of flashing fireflies \cite{sarfati2021}, neural synchrony \cite{Grinstein2005,varela2001,womelsdorf2014}, population dynamics in biology \cite{winfree1967}, computer systems \cite{korniss2003}, and many others. %
Whilst much early study focussed on the Kuramoto model \cite{kuramoto75,kuramoto84} with full coupling between the oscillators (or via mean field interactions), the study of synchronization is naturally generalized to complex networks of interaction between the oscillators (e.g. \cite{aren08} and references therein).
In this setting, each oscillator adjusts its activity in a decentralised fashion, as a function of the activity of the subset of oscillators that it is connected to via the complex network.
Whilst the capacity for synchronization of the oscillators, and the relative quality of such synchronization, varies with the strength of connectivity, it will clearly also vary with the underlying network structure \cite{aren08,porter2016}.
Establishing a full analytic relationship between network structure and synchronizability is the focus of this article.

Indeed, this relates to the more general research question of how dynamics on such complex networks relate to their structure \cite{porter2016}. Whilst the importance of such investigations have long been recognised, characterising the structure-dynamics relationship remains ``much less well understood'' than characterising structure alone \cite{mitchell2006} since the problem is simply more diverse \cite{barabasi2009}.
How synchronization relates to network structure between oscillators is seen of central importance to this wider structure-dynamics question \cite{porter2016,aren08}, because of the significant interdisciplinary and historical interest in the dynamic process of synchronization.

Formally, given a directed network structure and coupling strengths, we aim to determine whether the network is synchronizable and then to rank the relative quality of that synchronization (in terms of robustness to perturbations away from a synchronized state).
We will take the standard approach of examining linear systems, which can be considered as a linearisation around an attractor or as an approximation to the weakly-coupled near-linear regime, where general analytic insights are possible (see \cite{BAJ, jost01a, aren08, Pikovsky01, atay06a,hunt10a,hunt12a,korn07} and the references therein).
As we will describe in \secRef{measuringSync}, in this regime one can easily write down synchronization conditions in terms of the eigenvalues of the network coupling matrix or weighted adjacency matrix $C$ (or equivalently those of the system Laplacian $L$) \cite{atay06a,hunt12a,korn07}.
But since for linear systems broad ranges of coupling strengths can guarantee synchronization for any network structure \cite{atay06a}, it is useful to have insights into the \emph{relative} synchronizability of a network beyond simply answering whether it will synchronise or not \cite{korn07}.
\textit{Relative synchronizability} can be quantified as the average \emph{steady-state distance from synchronization} $\left\langle \sigma^2 \right\rangle$ under noise-driven dynamics (see \secRef{measuringSync}), having been computed analytically using all eigenvectors and all eigenvalues for symmetric $C$ \cite{korn07,hunt10a,hunt12a}. %
Heuristics based on leading eigenvalues or spread of eigenvalues are also often used \cite{korn07,nish10a,jost01a,alm07a}.

Yet despite the analytic success under these conditions, it could not be said that we have a systematic analytic understanding of how relative synchronizability relates to network \emph{structure}. %
This is because, apart from the extremal eigenvalues which are related to global properties \cite{chung97, bauerjost13}, there is no known general method to relate eigenvalues of the coupling matrix $C$ or the Laplacian $L$ to network structure.
Specifics are known only for very special cases, such as fully-connected and bipartite graphs,
yet for non-symmetric networks the corresponding coupling operator is not self-adjoint anymore and only in very special cases connections between eigenvalues and structure of the network are known \cite{bauer12}.
There are of course many studies empirically evaluating the dependence of synchronizability on network structure either based on the above analytic forms from the eigenvalue spectra \cite{korn07,hunt12a} or otherwise heuristic approximations \cite{korn07,nish10a,jost01a,alm07a,barahona2002,aguirre2014,nishikawa2003,nishikawa2017} or numerical simulations \cite{grabow2010,kim13a,kim2015,shanahan2008,nordenfelt2014,buscarino2013}, or else evaluation of the synchronizability of standalone motif structures outside of a network setting \cite{morenovega2004,lodato2007,li2010}.
Some results across these are relatively consistent, including that more random network structures tend to synchronize more effectively than regular ring or small-world structures \cite{nishikawa2003,atay04a,korniss2003,barahona2002,grabow2010,kim13a,kim2015,aren08}, and that degree homogeneity is helpful for synchronization \cite{korn07,nish10a,aren08}.
However Arenas et al. \cite{aren08} state that due to difficulties in isolating single network characteristics while controlling others, results can be conflicting, meaning that the situation remains ``quite confusing'' and can be misinterpreted.
Crucially, even where we can analytically compute a measure of synchronizability of a network, we cannot yet fully articulate the contribution of any local sub-network motif structures \cite{milo02a} such as feedback loops to promoting or detracting from whole of network synchronizability.

Here, we provide a full and general analytical solution to this problem, including quantifying how the different network motifs shape the synchronizability of the network as a whole.
Specifically, in \secRef{Covariance} and \secRef{motifRelationships}, we introduce a new method to measure the relative synchronizability via the steady-state distance from synchronization $\left\langle \sigma^2 \right\rangle$ of a network structure (with final results in \eq{synchronizabilityMotifsContinuous} and \eq{synchronizabilityMotifsDiscrete}).
In comparison to previous analytic work \cite{korn07,hunt12a}, our method determines relative synchronizability without requiring symmetry, nonnegativity, or diagonalizability of the weighted coupling matrix $C$. This is important theoretically in view of generality of results (e.g. the importance of diagnolizability \cite{nishikawa2006}), but particularly in view of biological applications, for instance, neuronal network models which are usually directed and use positive and negative weights corresponding to excitatory and inhibitory synapses, with these playing different roles in synchronization \cite{nish10a}. 
More importantly, our approach naturally provides the first analytic description relating relative synchronizability to local network structures, specifically ``process motifs'' as structured sets of walks on sub-network structures \cite{schwarze2021}.
We will show in \secRef{discussion} that networks become less synchronizable as the (weighted) sum of walks from each node that \textit{converge} to a common end-point rather than \textit{diverge} becomes larger.
We show subtle differences between discrete and continuous-time dynamics, where with discrete-time dynamics only convergent walks of the same length detract from synchronizability, whereas in continuous-time dynamics more general feedforward and feedback process motif structures also detract.
These results have important implications for our understanding of natural networks and their synchronizability, in particular regarding synchronisation phenomena in neural and cell-regulatory networks which exhibit a prevalence of such loop process motif structures \cite{milo02a}, in addition to directed and signed edge weights.
We expect that the new analytic forms we provide will be utilized to analytically explore a wide range of new relationships between network structure and relative synchronizability.

\section{Methods: Relative synchronizability}
\label{sec:measuringSync}

\subsection{Models of dynamics on a network structure}

Given the structure and edge weights $C$ of a network of $N$ nodes, we consider the dynamics on such networks in two very simple and widely-used models of linearly-coupled noise-driven systems.
These are: i. the continuous-time (differential equation) multivariate Ornstein-Uhlenbeck process \cite{uhl30,oks03,barn09b,schwarze2021} (also known as the Edwards-Wilkinson \cite{edw82} process on a network \cite{hunt12a,korn07}):
\begin{align}
	d\vec{x}(t) = - \vec{x}(t) (I - C) \theta dt + \zeta d\vec{w}(t)
	\label{eq:ornsteinUhlenbeck},
\end{align}
and ii. the discrete-time (difference equation) multivariate or vector autoregressive (VAR) process \cite{barn09b}:
\begin{align}
	\vec{x}(t + 1) = \vec{x}(t) C + \zeta \vec{r}(t)
	\label{eq:discreteARProcess},
\end{align}
on that network.
In both cases, $C=[C_{ji}]$ is the $N \times N$ \emph{connectivity matrix} (or \textit{weighted} adjacency matrix), where $C_{ji}$ is the (real, possibly negative) weight of the directed connection from node $j$ to $i$, and $I$ is the $N \times N$ identity matrix.
The current node values or activity levels $\vec{x}(t) = \{ x_1(t), x_2(t), ... x_N(t)\}$ is a row vector here.

These processes, with reversion rate $\theta > 0$ in continuous-time, are driven by uncorrelated noise terms with strength $\zeta^2$ (for a multivariate Wiener process $\vec{w}(t)$, and mean-zero unit-variance Gaussian noise $\vec{r}(t)$) with covariance matrix equal to $\zeta^2 I$ in both cases.
These noisy perturbations put the network into a perennial transient with respect to the synchronized state or \textit{zero mode} $\vec{\psi_0} = [ 1, 1, \ldots , 1 ]$ as a potential attractor.
Indeed, most studies (e.g. \cite{korn07,hunt10a,hunt12a,atay04a}) consider $C$ with $\vec{\psi_0}$ as an eigenvector with eigenvalue $\lambda_0=1$ (so the sum of all incoming edges to a target $\sum_{j}{C_{ji}} = 1$, i.e. with diffusive coupling), giving the possibility of a dynamic synchronization around $\vec{\psi_0}$ rather than trivial synchronization around the null state.
We will consider $C$ both with and without $\vec{\psi_0}$ as an eigenvector.
We also note that other studies often focus on the \textit{Laplacian} instead of the coupling matrix $C$, specifically here the normalized Laplacian \cite{StadlerBook} $L= I-C$ (when $\vec{\psi_0}$ is an eigenvector with eigenvalue $\lambda_0=1$).

\subsection{Synchronization conditions}
\label{sec:syncConditions}

It is well-known that the necessary \emph{synchronization conditions} 
(for the case without constant driving noise; demonstrated for constant driving noise in \cite{hunt12a})
can be stated in terms of the eigenvalues $\lambda_v$ of $C$ (or translated to those of $L$):
\begin{enumerate}
\item For continuous-time dynamics in \eq{ornsteinUhlenbeck}, we require $\mathrm{Re}(\lambda_v) < 1$ for all eigenvalues $\lambda_v$ of $C$ (except for $\lambda_0$ when $\vec{\psi_0}$ is an eigenvector).
\item For discrete-time dynamics in \eq{discreteARProcess}, we require $| \lambda_v | < 1$ for all $\lambda_v$ (again except for $\lambda_0$ when $\vec{\psi_0}$ is an eigenvector).
\end{enumerate}

Since these conditions can be met for any network structure by multiplicatively lowering the magnitudes of the $C_{ji}$ where required \cite{atay06a}\footnote{Though in lowering the eigenvalue of $\vec{\psi_0}$ this may result in trivial synchronization around the null state.}, it is useful to have insights into \emph{relative synchronizability} beyond simply answering whether the network will synchronise or not \cite{korn07}.
For various types of dynamics (not neccessarily linear systems, with or without time delays), one can make inferences regarding the relative synchronizability of various networks from their eigenvalue spectra \cite{atay05a} (ignoring $\lambda_0$ corresponding to $\vec{\psi_0}$). %
While the precise definition of synchronizability differs with dynamics, the general picture is that the relative synchronizability increases with the distance of the eigenvalues (in particular the largest) from the stability boundaries above, and also as the spread of eigenvalues becomes smaller \cite{atay05a,korn07,nish10a,jost01a,alm07a}.

\subsection{Steady-state distance from synchronisation $\left\langle \sigma^2 \right\rangle$}

Towards a more precise alternative, we can compute:
\begin{align}
	\sigma^2(t) = \frac{1}{N} \sum_{i=1}^N{ (x_i(t) - \bar{x}(t))^2 }
	\label{eq:degreeOfSync},
\end{align}
to measure the distance from synchronization at time $t$ in the noise-driven cases of \eqs{ornsteinUhlenbeck}{discreteARProcess}, with $\bar{x}(t)= \frac{1}{N}\sum_{i=1}^{N}{x_i(t)}$. %
Then, to evaluate relative synchronizability directly we compute $\left\langle \sigma^2 \right\rangle$ \cite{atay04a,korniss2003,korn07}, the expectation value which ``represents an average over the statistical ensemble'' \cite{barn09b} over the noise \cite{hunt10a}, dropping the time argument to indicate steady-state \cite{korn07}.
Under the standard synchronization conditions (above), this average \textit{steady-state} distance from the synchronized state ``approaches a finite value in the $t \rightarrow \infty$ limit'' \cite{hunt10a}, i.e. $\left\langle \sigma^2 \right\rangle < \infty$ \cite{hunt12a,hunt10a}, in contrast to $\left\langle \sigma^2 \right\rangle$ diverging when the synchronization conditions are not met.\footnote{In the absence of driving noise in \eqs{ornsteinUhlenbeck}{discreteARProcess}, the synchronization condition implies that
$\sigma^2(t) \rightarrow 0$ as $t \rightarrow \infty$ \cite{atay04a}.
This means that the network would completely synchronize to a common state without driving noise.
In other words, the synchronization condition is that the synchronized states are \textit{stable}; i.e. the system is stable to disturbances orthogonal to $\vec{\psi_0}$ \cite{korn07}.
In the presence of driving noise then, $\left\langle \sigma^2 \right\rangle < \infty$ implies that the network will maintain a perennial approach to synchronization \emph{without diverging}.
}
So, since $\left\langle \sigma^2 \right\rangle$ takes a \emph{range} of finite values when the network satisfies the synchronization condition \cite{korn07}, it provides a nuanced and meaningful indication of relative synchronizability, with smaller values indicating stronger relative synchronizability under driving noise. %
For practical use, one could compute $\left\langle \sigma^2(t) \right\rangle$ empirically via \eq{degreeOfSync} after some fixed number of time-steps (e.g. \cite{atay04a}), %
however we will focus on analytically evaluating $\left\langle \sigma^2 \right\rangle$.

We observe that:
\begin{align}
	\left\langle \sigma^2 \right\rangle & = \lim_{t \rightarrow \infty} \frac{1}{N} \sum_{i=1}^{N}{ \left\langle (x_i(t) - \bar{x}(t))^2 \right\rangle }, \label{eq:syncInTermsOfComponents} \\
		& = \lim_{t \rightarrow \infty} \frac{1}{N} \mathrm{trace}\left( \left\langle (U^T \vec{x}(t)^T) (\vec{x}(t) U) \right\rangle \right), \nonumber \\
		& = \frac{1}{N} \mathrm{trace}\left( U^T \Omega U \right), \nonumber \\
		& = \frac{1}{N} \mathrm{trace}\left( \Omega_U \right)
	\label{eq:covarianceInProjected},
\end{align}
where $U$ is the \emph{centering matrix} or \textit{unaveraging operator}\footnote{See further definition and useful expressions in \app{traceInOrthogonalSpace}.} $\vec{x}U = \vec{x} - \bar{x}\vec{\psi_0}$ \cite{mard95}; whilst $\Omega = \lim_{t \rightarrow \infty} \left\langle \vec{x}(t)^T \vec{x}(t) \right\rangle$ (with elements $\Omega_{ij} = \lim_{t \rightarrow \infty} \left\langle x_i(t) x_j(t) \right\rangle$) is the steady-state symmetric covariance matrix of $\vec{x}(t)$; 
and $\Omega_U \triangleq \lim_{t \rightarrow \infty} \left\langle U^T \vec{x}(t)^T \vec{x}(t) U \right\rangle$, and $ \Omega_U = U^T \Omega U$ when $\Omega$ is well-defined.
Note that $U$ makes projections onto the space orthogonal to the zero-mode $\vec{\psi_0}$, so we can interpret \eq{covarianceInProjected} as measuring $\left\langle \sigma^2 \right\rangle$ proportional to the trace of $\Omega$ in the space orthogonal to $\vec{\psi_0}$.

To analytically evaluate $\left\langle \sigma^2 \right\rangle$ then, we must determine $\Omega_U$ or $\mathrm{trace}\left( \Omega_U \right)$ directly from $C$ (or its eigenvalues).

\subsection{Existing methods to determine $\left\langle \sigma^2 \right\rangle$}
\label{sec:existingMethod}

Galan \cite{gal08a} presents an eigenvalue decomposition technique for computing $\Omega$ for a generalized case of \eq{ornsteinUhlenbeck} directly from the eigenvectors and eigenvalues of $C$.\footnote{These calculations were made for a more general case of our continuous-time system \eq{ornsteinUhlenbeck} with correlated noise. However the solution presented is actually for a discretization corresponding to \eq{discreteARProcess}, which no longer corresponds to the continuous-time equation (as shown by Barnett et al. \cite{barn09b}).}
However, this approach as presented is not applicable to $C$ in the common case (for synchronization) when it has $\lambda_0=1$ for eigenvector $\vec{\psi_0}$, and it also requires $C$ to be diagonalizable.

Korniss \cite{korn07} computes $\left\langle \sigma^2 \right\rangle$ for a continuous-time system (akin to \eq{ornsteinUhlenbeck}) using a Green's functions approach to the covariance components.
(When mapped to our \eq{ornsteinUhlenbeck}) this approach requires $C$ to be symmetric, requires positive edge weights $C_{ij}>0$, and assumes the common case that $\vec{\psi_0}$ is an eigenvector of $C$ with $\lambda_0 = 1$.
Using the well-known eigenvalue expansion for the Green's function (or two point correlation function) $\left\langle \sigma^2 \right\rangle$ is computed for such $C$ with $\theta=\zeta=1$ as:\footnote{We have mapped the solution to eigenvalues of $C$ rather than $L$, and for noise $\vec{w}(t)$ with the same covariance we consider here (Korniss \cite{korn07} considered noise with twice as much covariance).}
\begin{align}
	\left\langle \sigma^2 \right\rangle = \frac{1}{2N} \sum_{v=1}^{N-1}{ \frac{1}{1 - \lambda_{v}} }
	\label{eq:widthAsInverseEigs},
\end{align}
where $\lambda_{v}$ are the eigenvalues of $C$ except for $\lambda_0$.
Insights were then derived into specific network structures \cite{korn07}, i.e.: for fixed total edge costs, fully-connected network of identical edge-weights minimize $\left\langle \sigma^2 \right\rangle$;
while for fixed edge weights and a fixed average degree, a perfectly homogeneous random graph where each node connects to exactly $k$ others is optimal.
Hunt et al. \cite{hunt12a} extended these results to compute $\left\langle \sigma^2 \right\rangle$ for
a coupling delay.

\section{Results: Analytic determination of steady-state distance from synchronisation $\left\langle \sigma^2 \right\rangle$}
\label{sec:Covariance}

In this section, we first describe a new method to compute the projected covariance matrix $\Omega_U$ for \eqs{ornsteinUhlenbeck}{discreteARProcess} directly from $C$ in \fullSubSecRef{Covariance}{powerSeriesMethod}.
Since this method, unlike previous work, does not require either symmetric or diagonalizable $C$, we are thus able to compute the steady-state distance from synchronisation $\left\langle \sigma^2 \right\rangle$ without such restrictions in \fullSubSecRef{Covariance}{solution} for the first time.
Numerical validation of the method is provided in \fullSubSecRef{Covariance}{numericalValidation}.
Perhaps more importantly than the result for $\left\langle \sigma^2 \right\rangle$ itself though, our method paves the way to directly reveal the contribution of motif structures on $\left\langle \sigma^2 \right\rangle$ in the following \secRef{motifRelationships}.

\subsection{Power series method for $\Omega_U$}
\label{sec:powerSeriesMethod}

First, we contribute a new approach to computing $\Omega_U$, which does not require symmetry or diagonalizability of $C$.
This approach is based on Barnett et al.'s method \cite{barn09b} to analytically obtain $\Omega$ as a power series of $C$.\footnote{Ref. \cite{barn09b} corrects an original attempt by Tononi et al.\cite{ton94}. Similar methods are also presented in \cite{pernice2011,hu2014}.}

For the \textbf{continuous-time process} with $\theta=\zeta=1$ in \eq{ornsteinUhlenbeck} Barnett et al. \cite{barn09b} show that the covariance matrix $\Omega$ is obtained via a power series \cite{barn09b}:
\begin{align}
	2 \Omega & = I + \frac{1}{2} (C^T + C) + \frac{1}{4} \left[ (C^2)^T + 2C^TC + C^2 \right] + \ldots
	, \nonumber \\
	& = \sum_{m=0}^{\infty}{ 2^{-m} \sum_{u=0}^{m}{\binom{m}{u} (C^u)^T C^{m-u}} }
	\label{eq:covarianceGeneralContinuous}.
\end{align}
The validity of this solution relies on both stationarity of \eq{ornsteinUhlenbeck} (given by $\mathrm{Re}(\lambda) < 1$, for all real parts  $\mathrm{Re}(\lambda)$ of eigenvalues $\lambda$ of $C$) and convergence of \eq{covarianceGeneralContinuous} (shown only for the stronger condition of $\rho (C) < 1$ for the spectral radius $\rho (C)$ of $C$ in \cite{barn09b}).

As such, this solution does not directly apply for our general case with $\lambda_0 = 1$ for the zero-mode $\vec{\psi_0}$, which is important for non-trivial synchronized solutions.\footnote{\label{fn:UbracketsOmega}One cannot trivially compute $\Omega_U = U^T \Omega U$ from \eq{covarianceGeneralContinuous} by relying on non-existent convergence of $\Omega$.}
To address this, we have extended the derivations of Barnett et al. \cite{barn09b} and Schwarze and Porter \cite{schwarze2021} in this case to examine $\Omega_U$ (in the space orthogonal to the zero-mode), demonstrating in \app{convergenceUTOmegaUContinuous} (c.f. \eq{bracketedCovarianceGeneralContinuousFinal}):
\begin{align}
	\Omega_U = \frac{\zeta^2}{2\theta} \sum_{m=0}^{\infty}{ 2^{-m} \sum_{u=0}^{m}{\binom{m}{u} U (C^u)^T C^{m-u} U} }
	\label{eq:bracketedCovarianceGeneralContinuousMain},
\end{align}
and showing the validity of this solution for $\rho (C U) < 1$, i.e. for $|\lambda_C|<1$ for all eigenvalues $\lambda_C$ of $C$ except that corresponding to $\vec{\psi_0}$.

Clearly this matches $\Omega_U=U^T \Omega U$ with $\theta=\zeta=1$ obtained directly from $\Omega$ in \eq{covarianceGeneralContinuous} for the more restricted regime where \eq{covarianceGeneralContinuous} is valid.
Intuitively we can view this as a projection of the covariance matrix onto the space orthogonal to the synchronized state $\vec{\psi_0}$.
We emphasise though (as per \fn{UbracketsOmega}) that one cannot simply write this down when \eq{covarianceGeneralContinuous} is not valid; the foundation provided by the proof in \app{convergenceUTOmegaUContinuous} is crucial.

For the \textbf{discrete-time process} with $\zeta=1$ in \eq{discreteARProcess}, Barnett et al. show that the general solution for $\Omega$ is obtained from:
\begin{align}
	\Omega = I + C^T C + (C^2)^T C^2 + \ldots = \sum_{u=0}^{\infty}{(C^u)^T C^u}
	\label{eq:covarianceGeneralDiscrete}.
\end{align}
Again, the validity of this solution requires stationarity of \eq{discreteARProcess} and convergence of \eq{covarianceGeneralDiscrete} (given by $|\lambda_C| < 1$ for both).

As such, this solution does not directly apply for our general case with $\lambda_0 = 1$ for the zero-mode $\vec{\psi_0}$.
To address this, we again extend the derivation in \cite{barn09b} to examine $\Omega_U$, demonstrating in \app{convergenceUTOmegaUDiscrete} (c.f. \eq{bracketedCovarianceGeneralDiscreteFinal}):
\begin{align}
	\Omega_U = \zeta^2 \sum_{u=0}^{\infty}{U (C^u)^T C^{u} U}
	\label{eq:bracketedCovarianceGeneralDiscreteMain},
\end{align}
and showing there the validity of this solution for $\rho (C U) < 1$.
Again, this clearly this matches $\Omega_U=U^T \Omega U$ with $\zeta=1$ obtained directly from $\Omega$ in \eq{covarianceGeneralDiscrete} for the more restricted regime where \eq{covarianceGeneralDiscrete} is valid.

\subsection{Solution for steady-state distance from synchronisation $\left\langle \sigma^2 \right\rangle$}
\label{sec:solution}

With the above solutions for $\Omega_U$, we can now provide the solutions for $\left\langle \sigma^2 \right\rangle$.

If we have $|\lambda_C| < 1$, for all eigenvalues $\lambda_C$ of $C$, then both \eq{covarianceGeneralContinuous} and \eq{covarianceGeneralDiscrete} are valid, and we can immediately evaluate \eq{covarianceInProjected} with $\Omega$ to get $\left\langle \sigma^2 \right\rangle$.
Alternatively, in the more interesting case with the zero-mode $\vec{\psi_0}$ with $\lambda_0 = 1$ (where non-trivial synchronized solutions are possible), then with $|\lambda_C| < 1$ for all eigenvalues $\lambda_C$ of $C$ except that corresponding to $\vec{\psi_0}$, \eq{bracketedCovarianceGeneralContinuousMain} and \eq{bracketedCovarianceGeneralDiscreteMain} are valid, and we can evaluate \eq{covarianceInProjected} from these.
Both approaches lead to identical forms for $\Omega_U$ and $\left\langle \sigma^2 \right\rangle$.

As such, for the \textbf{continuous-time process} in \eq{ornsteinUhlenbeck}, we expand \eq{covarianceInProjected} using \eq{bracketedCovarianceGeneralContinuousMain} to obtain the solution for steady-state distance from sync:
\begin{align}
	\left\langle \sigma^2 \right\rangle = \frac{\zeta^2}{2\theta} \sum_{m=0}^{\infty} & \frac{2^{-m}}{N} \sum_{u=0}^{m} \binom{m}{u} \mathrm{trace}\left( U (C^u)^T C^{m-u} U \right).
	\label{eq:synchronizabilityContinuous}
\end{align}

Similarly, for the \textbf{discrete-time process} in \eq{discreteARProcess}, we expand \eq{covarianceInProjected} using \eq{bracketedCovarianceGeneralDiscreteMain} to obtain the corresponding solution for steady-state distance from sync:
\begin{align}
	\left\langle \sigma^2 \right\rangle & = \frac{\zeta^2}{N} \sum_{u=0}^{\infty}{ \mathrm{trace}\left( U (C^u)^T C^{u} U \right)}.
	\label{eq:synchronizabilityDiscrete}
\end{align}

These solutions represent a substantial advance over existing methods, in neither requiring symmetric nor diagonalizable $C$, and we are thus able to compute the steady-state distance from synchronisation $\left\langle \sigma^2 \right\rangle$ without such restrictions for the first time.
Strictly, we do require $|\lambda_C| < 1$ for all eigenvalues $\lambda_C$ of $C$ except that corresponding to $\vec{\psi_0}$, though this is equivalent to the synchronizability condition in discrete time (\fullSubSecRef{measuringSync}{syncConditions}).

\subsection{Numerical validation}
\label{sec:numericalValidation}

With the analytic solutions in place in \eqs{synchronizabilityContinuous}{synchronizabilityDiscrete}, we seek to provide numerical validation for them.
Matlab code for these experiments is distributed in the \textit{linsync} toolbox \cite{linsync}, including adaptations of \textit{n\_comp} \cite{ncomp}.

In order to check the validity of the solution over a wide variety of networks, we generate sample network structures throughout a small-world transition on an $N$-node Watts-Strogatz ring network \cite{watts98}.
Initially each node has $d$ \textit{directed} incoming edges from its closest neighbors on the ring, then sources of each edge are randomised with probability $p$.
The $d$ incoming edges for each node $i$ have equal edge weights $C_{ji} = c / d$, with total incoming (non-self) connection weights $c = \sum_{j \neq i} C_{ji}$, and each node has self-link weights $C_{ii} = 1 - c$ (which are not randomised).
This setup with $\sum_{j}{C_{ji}} = 1$ ensures that each connectivity matrix $C$ has zero-mode eigenvector $\psi_0$ with $\lambda_0 = 1$, facilitating dynamic synchronization around $\psi_0$.
In the few cases where the edge randomization process results in networks which do not meet the the requirement of $|\lambda_C| < 1$ for all other eigenvalues $\lambda_C$ (and by implication the synchronization conditions in \fullSubSecRef{measuringSync}{syncConditions}), e.g. by disconnecting the network, these network realisations are discarded.

Numerical results are generated for the continuous-time dynamics in \eq{ornsteinUhlenbeck} with $\theta = \zeta = 1$.
We compute the steady-state distance from synchronisation $\left\langle \sigma^2 \right\rangle$ analytically via \eq{synchronizabilityContinuous}.
We also compute this quantity empirically from time series simulated from \eq{ornsteinUhlenbeck} using the exact method (following Appendix A of Barnett et al. \cite{barn09b}) with samples taken every $dt$=1.0 time units, and thus compute the average $\left\langle \sigma^2 \right\rangle_E$ over $L$ time-series samples via \eq{degreeOfSync}.

\fig{convergence} plots the mean of the relative absolute error $\left| \left\langle \sigma^2 \right\rangle - \left\langle \sigma^2 \right\rangle_E \right| / \left\langle \sigma^2 \right\rangle$ for parameters $N=100$, $d=4$ and $c=0.5$, over various values of network randomisation parameter $p$, and various sampling lengths $L$. For each parameter combination, the mean is taken over 2000 network realisations.

Crucially, this demonstrates that the analytic results are well validated against the empirical experiments, because the empirical results converge exponentially to the expected analytic values as the number of samples (which $\left\langle \sigma^2 \right\rangle_E$ is averaged over) increases.
Moreover, this convergence is observed across all values of network randomisation parameter $p$, from regular ring to small-world to random networks as $p$ increases from 0 to 1.
We also see that there is a smaller relative error for random networks for the same number of samples $L$; we will explain why this is the case after we analytically relate synchronizability to network structure in the next section.

\begin{figure}
		\includegraphics[width=0.95\columnwidth]{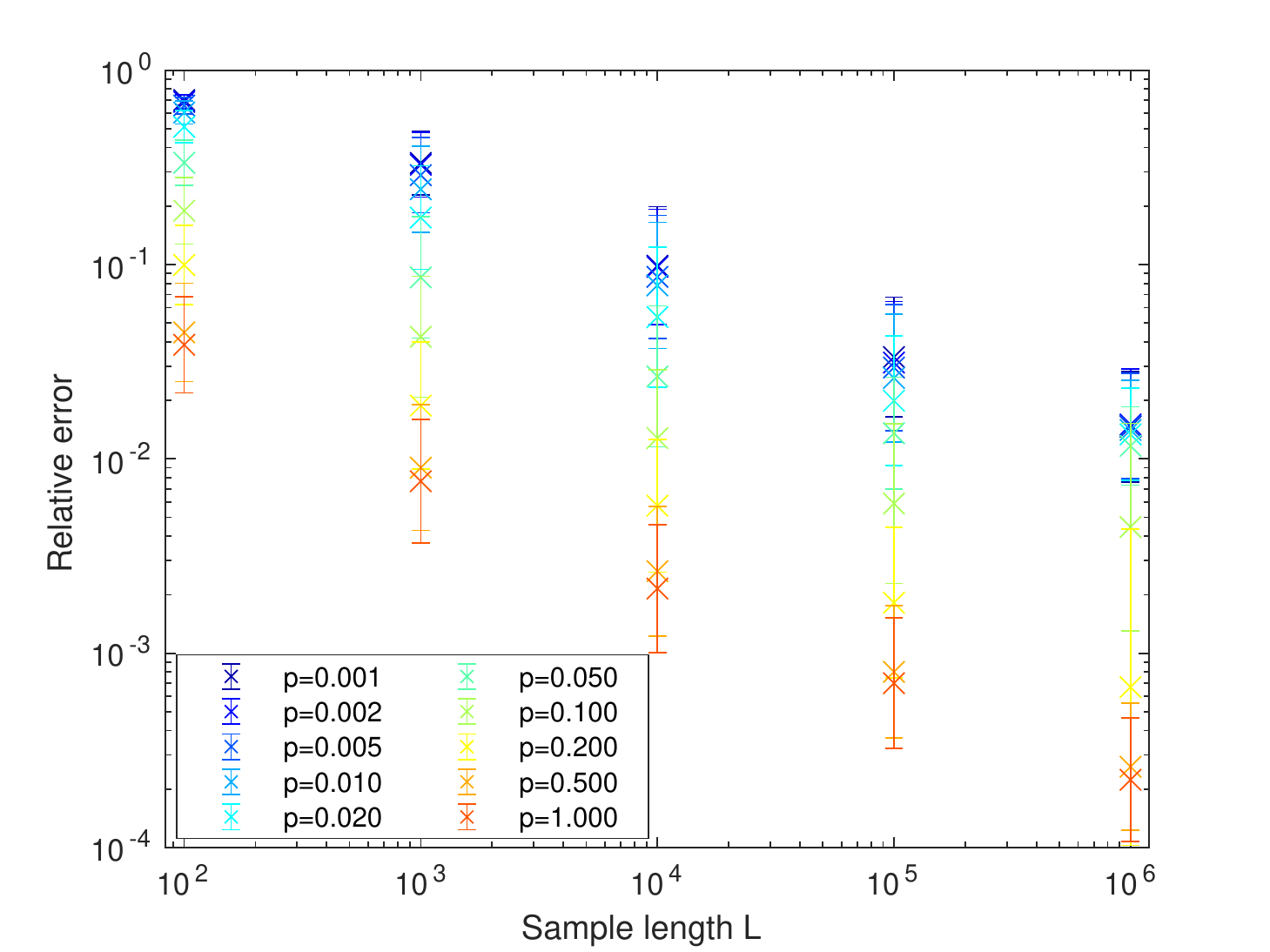}
	\caption{\label{fig:convergence} \textit{Numerical results validating} our analytic expression for average steady-state distance from synchronisation $\left\langle \sigma^2 \right\rangle$ for continuous-time dynamics against empirical measurements, throughout a small-world transition on an $N=100$ ring network with network randomization parameter $p$.
	Parameters for the network connectivity matrix are described in \fullSubSecRef{Covariance}{numericalValidation}, with $d=4$ and $c=0.5$.
	The plot shows the relative absolute error of the empirical measurements $\left\langle \sigma^2 \right\rangle_E$, averaged over 2000 network realisations, versus the number of time series samples $L$ for the empirical measurements for each network.
	Error bars represent the standard deviation across realisations in log space.}
\end{figure}

\section{Results: Relationship of network motifs to synchronizability}
\label{sec:motifRelationships}
In the spirit of Barnett et al.'s~\cite{barn09b,barn11a} analysis of the Tononi-Sporns-Edelman (TSE) complexity \cite{ton94} in networks (since inspiring or mirrored by studies of several other measures of dynamics on networks \cite{pernice2011,lizier2012storageLoops,hu2014,novelli2020,schwarze2021}), we now use the power-series expansions in \eqs{synchronizabilityContinuous}{synchronizabilityDiscrete} to directly reveal the roles of motif structures in $\left\langle \sigma^2 \right\rangle$.

First, we define the following notation.
We use $\walkmotif{a}{b}{M}$ to represent the \textit{weighted walk motif count} of all directed walks from node $a$ to $b$ of length $M$:
\begin{align}
	\walkmotif{a}{b}{M} & = \sum_{n_1,n_2 \ldots n_{M-1}}{C_{an_1}C_{n_1n_2} \ldots C_{n_{M-1}b}}, \nonumber \\
	& = (C^M)_{ab} = ((C^M)^T)_{ba},
\end{align}
and for shorthand write $C^M_{ab}$ for $(C^M)_{ab}$ as the $(a,b)$ entry in $C^M$ (as distinct from the edge weight $C_{ab}$ raised to the power $M$).

We also use $\walkprod{a}{b}{M_1}{a}{e}{M_2}$ to represent \textit{weighted dual walk motif counts} as the product of the two weighted walk motif counts starting from node $a$ and finishing at nodes $b$ and $e$ respectively:
\begin{align}
	\walkprod{a}{b}{M_1}{a}{e}{M_2} & = \walkmotif{a}{b}{M_1}\walkmotif{a}{e}{M_2}, \nonumber \\
	& = (C^{M_1})_{ab}(C^{M_2})_{ae}, \\
	& = ((C^{M_1})^T)_{ba}(C^{M_2})_{ae}
	\label{eq:weightedDualWalkMotifCounts}.
\end{align}
The sub-case of $\walkprod{a}{b}{M_1}{a}{b}{M_2}$ (i.e. where the walks start at node $a$ and end at node $b$ in both cases) are referred to as \textit{weighted closed dual walk motif counts}, representing the weighted counts of pairs of walks starting at node $a$ and ending at node $b$ via walks of length $M_1$ and $M_2$ respectively.

By convention, $\walkmotif{a}{b}{0} := \delta_{ab}$, so $\walkprod{a}{b}{0}{a}{e}{M_2} = \walkmotif{a}{e}{M_2}$ when $a=b$.
Several examples of walk motifs included in these counts are shown in \fig{motifsForContinuous} and \fig{motifsForDiscrete}.
Both $\walkmotif{a}{b}{M}$ and $\walkprod{a}{b}{M_1}{a}{e}{M_2}$ are examples of \textit{weighted process motif counts} \cite{schwarze2021}, being counts of structured sets of walks on a network, and we note that many such process motifs or walks can occur on a given underlying structural motif \cite{schwarze2021}.

Finally, note that our expansions will utilize the key result that the trace of a given matrix $A$ after the centering projections are applied to produce $\mathrm{trace}(U^TAU)$ is given by:
\begin{equation} \mathrm{trace}(U^TAU) = \sum_i{A_{ii}} - \frac{1}{N} \sum_{i,j}{ A_{ij} },
\label{eq:traceInOrthogonalBookend}
\end{equation}
as proven in \lemmaRef{traceInOrthogonalBookend} in \app{traceInOrthogonalSpace}.

\subsection{Continuous-time processes}
\label{sec:synchronizabilityMotifsContinuous}

Now, expanding \eq{synchronizabilityContinuous} using \eq{traceInOrthogonalBookend}, the steady-state distance from sync becomes:
\begin{align}
	\left\langle \sigma^2 \right\rangle = & \frac{\zeta^2}{2\theta} \sum_{m=0}^{\infty} \frac{2^{-m}}{N} \sum_{u=0}^{m} \binom{m}{u} \times \nonumber \\
	 & \ \ \left( \sum_{i}{ ((C^u)^T C^{m-u})_{ii}  } \nonumber - \frac{1}{N}\sum_{i,j}{ ((C^u)^T C^{m-u})_{ij} } \right),
\end{align}
so:
\begin{align}
	\left\langle \sigma^2 \right\rangle = \frac{\zeta^2}{2\theta} \sum_{m=0}^{\infty} & \frac{2^{-m}}{N} \sum_{u=0}^{m} \binom{m}{u} \times \nonumber \\
				& \left( \sum_{i,k}{ C^u_{ki} C^{m-u}_{ki} } - \frac{1}{N} \sum_{i,j,k}{ C^u_{ki} C^{m-u}_{kj} } \right), \\
	\left\langle \sigma^2 \right\rangle =  
			\frac{\zeta^2}{2\theta} \sum_{m=0}^{\infty} & \frac{2^{-m}}{N} \sum_{u=0}^{m} \binom{m}{u} \times \nonumber \\
				& \sum_{i,k}{
			\left( \walkprod{k}{i}{u}{k}{i}{m-u} - \frac{1}{N} \sum_{j}{ \walkprod{k}{i}{u}{k}{j}{m - u} } \right) },
	\label{eq:synchronizabilityMotifsContinuous}
\end{align}
using \eq{weightedDualWalkMotifCounts}.
 
We see that the steady-state distance from synchronization $\left\langle \sigma^2 \right\rangle$ is determined by weighted dual walk motif counts, as the difference between contributions of:
\begin{enumerate}
	\item \textit{closed dual walk motif} counts $\walkprod{k}{i}{u}{k}{i}{m-u}$, being in general \textit{feedforward loop motifs} from source nodes $k$ to targets $i$ in lengths $u$ and $m-u$, including the sub-cases of:
	\begin{enumerate}
		\item \textit{feedback loop motifs} via $\walkprod{i}{i}{u}{i}{i}{m-u}$ (where $k= i$) as the motif of two feedback loops on nodes $i$ of cycle lengths $u$ and $m-u$ (or a single feedback loops $\walkmotif{i}{i}{m}$ of cycle length $m$ when $u = 0$ or $m$); and
	\end{enumerate}
	\item \textit{all dual walk motif} counts $\walkprod{k}{i}{u}{k}{j}{m - u}$ from sources $k$ to nodes $i$ and $j$ over (potentially different) lengths $u$ and $m-u$, averaged over all $j$.
\end{enumerate}
Some sample low-order closed and other dual walk motifs contributing here to $\left\langle \sigma^2 \right\rangle$ are shown in \fig{motifsForContinuous}.

We will discuss the interpretation of this result in \secRef{discussion}.

\begin{figure}
		\subfigure[$\walkprod{k}{i}{u}{k}{i}{m-u}$]{\label{fig:k_i_u_k_i_m-u}\includegraphics[height=\generalMotifClosedHeight]{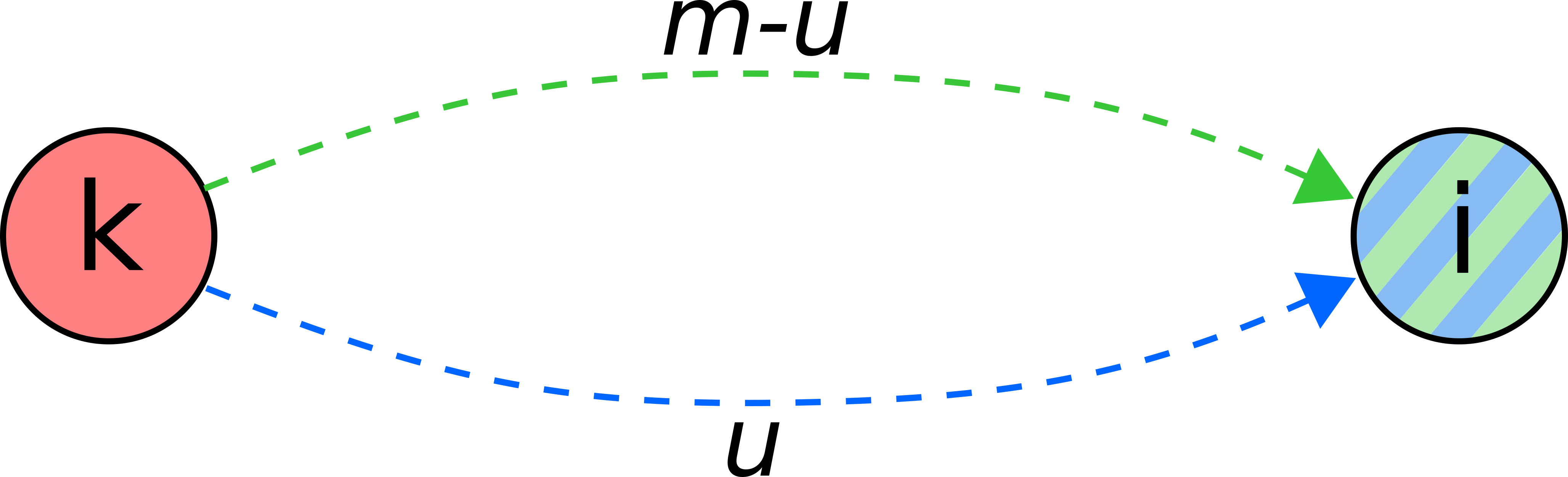}} {  }
		\subfigure[$\walkprod{k}{i}{u}{k}{j}{m-u}$]{\label{fig:k_i_u_k_j_m-u}\includegraphics[height=\generalMotifOpenHeight]{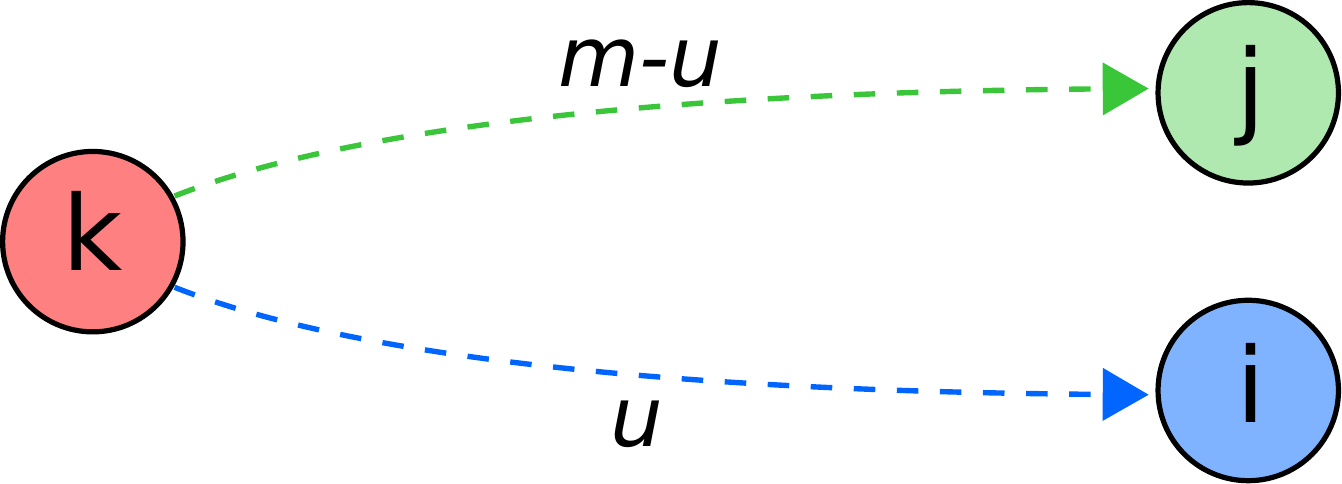}} {  }
		\\
		\subfigure[$\walkmotif{i}{i}{2}$]{\label{fig:i_i_2}\includegraphics[width=\motifFigWidth]{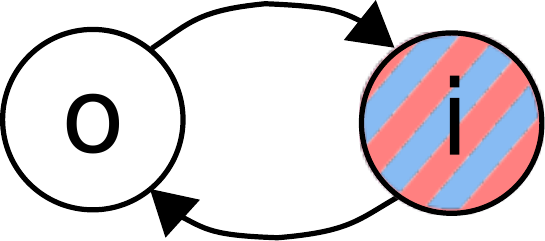}} {  }
		\subfigure[$\walkmotif{i}{i}{3}$]{\label{fig:i_i_3}\includegraphics[height=\motifFigHeight]{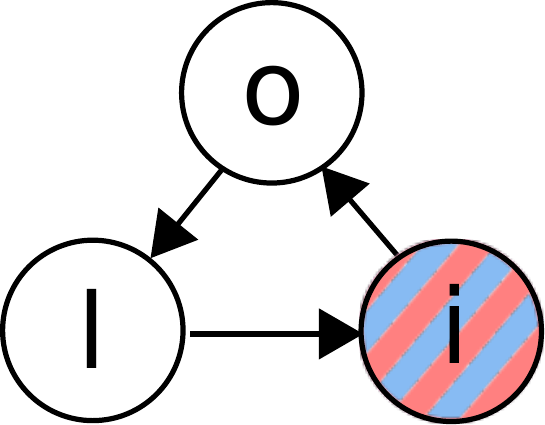}} {  }
		\subfigure[$\walkprod{k}{i}{1}{k}{i}{2}$]{\label{fig:k_i_1_k_i_2}\includegraphics[height=\motifFigHeight]{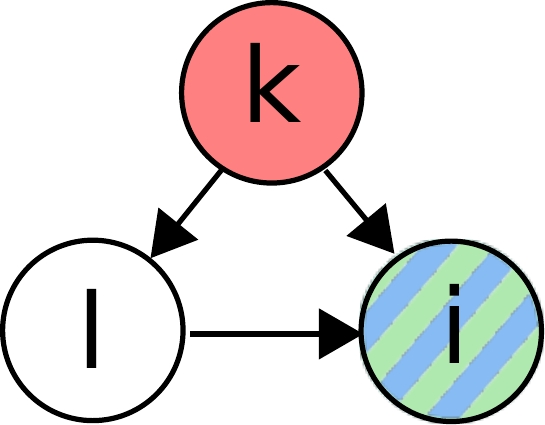}} {  }
		\subfigure[$\walkprod{k}{i}{1}{k}{j}{2}$]{\label{fig:k_i_1_k_j_2}\includegraphics[height=\motifFigHeight]{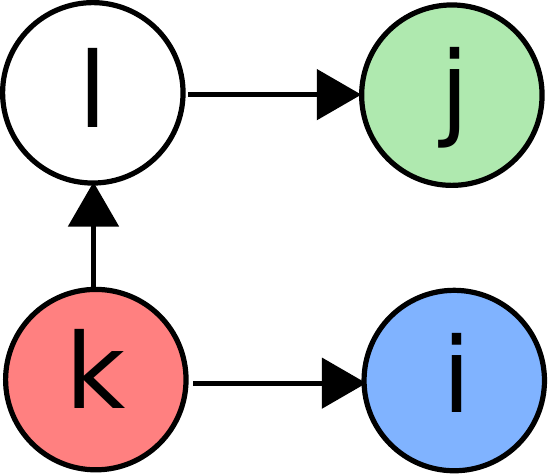}} {  }
	\caption{\label{fig:motifsForContinuous} Process motifs implicated in calculation of $\left\langle \sigma^2 \right\rangle$ for the \emph{continuous-time} process, being the difference between weighted counts of \subref{fig:k_i_u_k_i_m-u} closed dual walk motifs $\walkprod{k}{i}{u}{k}{i}{m-u}$ (with two walks from node $k$ to $i$ in $m-u$ and $u$ steps respectively) and \subref{fig:k_i_u_k_j_m-u} all dual walk motifs $\walkprod{k}{i}{u}{k}{j}{m-u}$ (with a walk from node $k$ to $i$ in $u$ steps and a walk to $j$ in $m-u$ steps respectively) averaged over $j$.
	Samples of process motifs these involve include: \emph{feedback loops} \subref{fig:i_i_2} and \subref{fig:i_i_3}, \emph{feedforward loops} \subref{fig:k_i_1_k_i_2},
	and general open \emph{dual walk motifs} where $j \neq i$ \subref{fig:k_i_1_k_j_2},
	as well as all of those shown for the discrete-time process in \fig{motifsForDiscrete}.
	Subfigure labels indicate the weighted counts to which the displayed process motifs contribute.
	Nodes are coloured as per $k$, $i$ and $j$ for $\walkprod{k}{i}{u}{k}{j}{m-u}$ in \subref{fig:k_i_u_k_j_m-u}, with cross-hatching on nodes indicating node $k = i$ in \subref{fig:i_i_2} and \subref{fig:i_i_3}, and node $j = i$ in \subref{fig:k_i_u_k_i_m-u} and \subref{fig:k_i_1_k_i_2}.}
\end{figure}

\subsection{Discrete-time processes}

Expanding \eq{synchronizabilityDiscrete} in a similar way to the continuous process (see full derivation in \app{discreteExpansion}) we obtain:
\begin{align}
	\left\langle \sigma^2 \right\rangle & = \frac{\zeta^2}{N} \sum_{u=0}^{\infty}{\left( \sum_{i,k}{ \walkprod{k}{i}{u}{k}{i}{u}	} - \frac{1}{N}\sum_{i,j,k}{ \walkprod{k}{i}{u}{k}{j}{u}	} \right)}, \nonumber \\
	& = \zeta^2(1-\frac{1}{N} ) + \frac{\zeta^2}{N} \sum_{u=1}^{\infty}{ \sum_{i,k}{ \left( \walkprod{k}{i}{u}{k}{i}{u} - \frac{1}{N} \sum_{j}{ \walkprod{k}{i}{u}{k}{j}{u} }  \right) } }
	\label{eq:synchronizabilityMotifsDiscrete}.
\end{align}

We see that $\left\langle \sigma^2 \right\rangle$ here is determined similarly from weighted dual walk motif counts, this time as the difference between contributions of:
\begin{enumerate}
	\item \textit{closed dual walk motif} counts $\walkprod{k}{i}{u}{k}{i}{u}$ from source nodes $k$ to targets $i$ in two walks of the same length $u$, including the sub-cases of:
	\begin{enumerate}
		\item \textit{feedback loop motifs} via $\walkprod{i}{i}{u}{i}{i}{u}$ (where $k=i$) as the motif of two \textit{feedback loops} on node $i$ of the same cycle lengths $u \geq 1$; and
	\end{enumerate}
	\item all \textit{dual walk motif counts} $\walkprod{k}{i}{u}{k}{j}{u}$ from sources $k$ to nodes $i$ and $j$ over the same length $u$, averaged over all $j$.
\end{enumerate}
Samples of process motifs contributing to the lowest-order closed
and other
dual walk motif counts relevant to $\left\langle \sigma^2 \right\rangle$ here are shown in \fig{motifsForDiscrete}.

\begin{figure}
		\subfigure[$\walkprod{k}{i}{u}{k}{i}{u}$]{\label{fig:k_i_u_k_i_u}\includegraphics[height=\generalMotifClosedHeight]{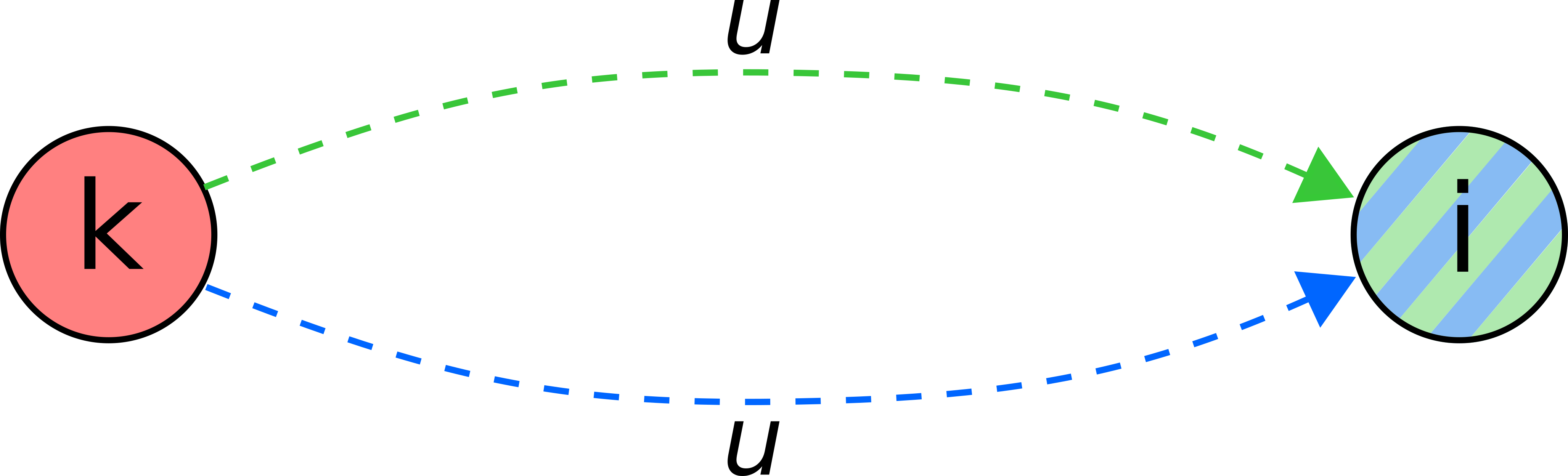}} {  }
		\subfigure[$\walkprod{k}{i}{u}{k}{j}{u}$]{\label{fig:k_i_u_k_j_u}\includegraphics[height=\generalMotifOpenHeight]{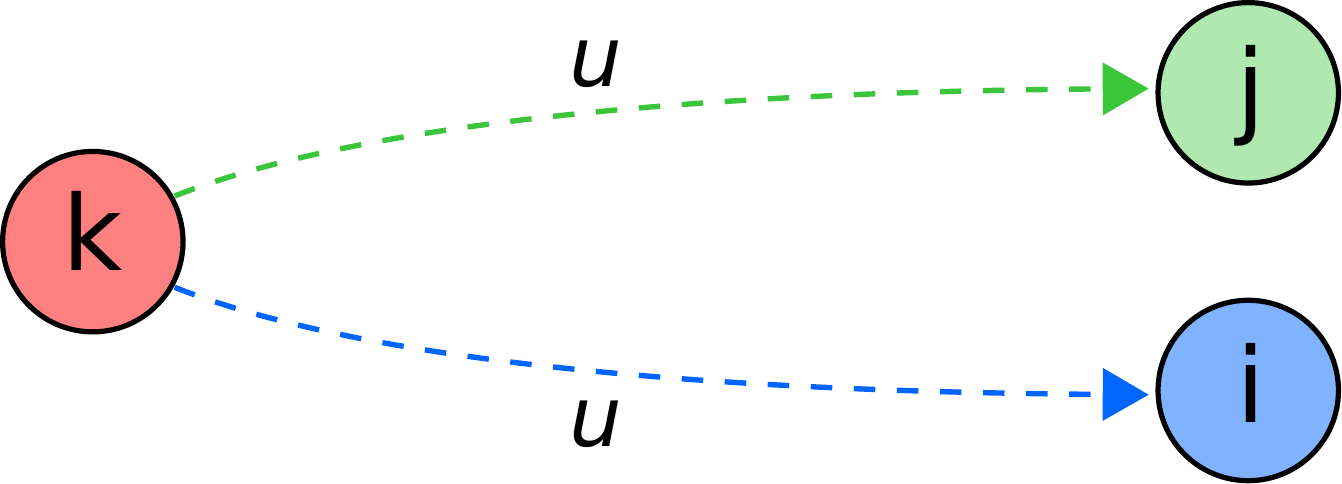}} {  }
		\\
		\subfigure[$\walkprod{k}{i}{2}{k}{i}{2}$]{\label{fig:k_i_2_k_i_2}\includegraphics[height=\motifFigHeight]{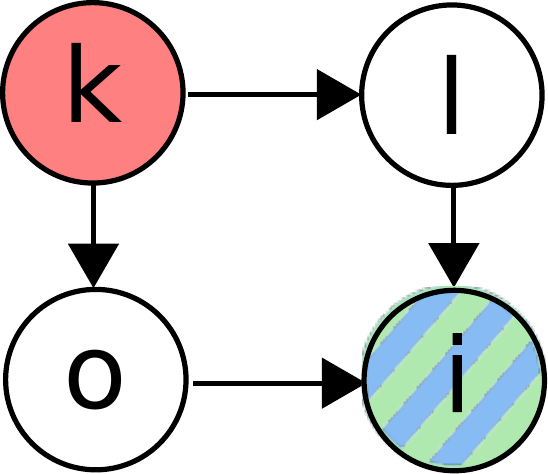}} {  }
		\subfigure[$\walkprod{k}{i}{1}{k}{j}{1}$]{\label{fig:k_i_1_k_j_1}\includegraphics[height=\motifFigHeight]{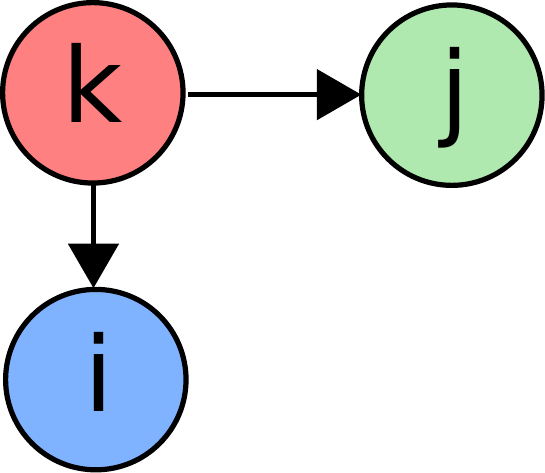}} {  }
		\subfigure[$\walkprod{k}{i}{2}{k}{j}{2}$]{\label{fig:k_i_2_k_j_2}\includegraphics[height=\motifFigHeight]{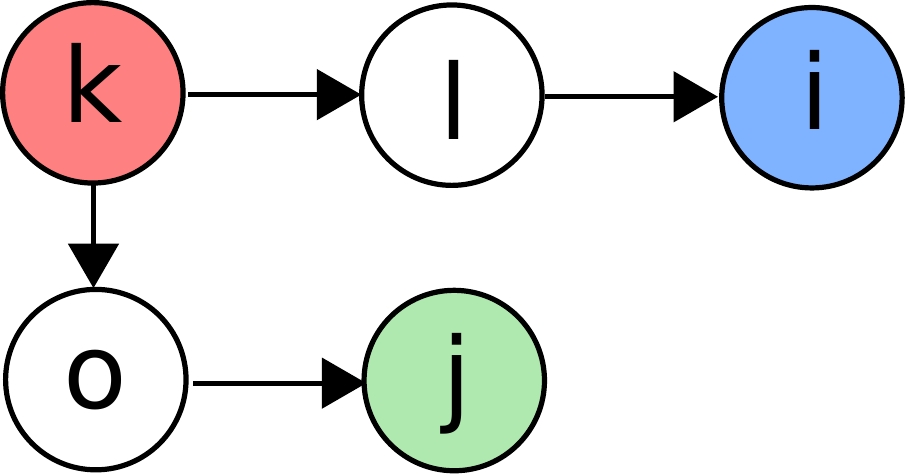}} {  }
	\caption{\label{fig:motifsForDiscrete} Process motifs implicated in calculation of $\left\langle \sigma^2 \right\rangle$ for the \textit{discrete-time} process, being the difference between weighted counts of \subref{fig:k_i_u_k_i_u} closed dual walk motifs $\walkprod{k}{i}{u}{k}{i}{u}$ and \subref{fig:k_i_u_k_j_u} averages of all dual walk motifs $\walkprod{k}{i}{u}{k}{j}{u}$ (with the dual walks of the \textit{same length} $u$ in both cases).
	Samples of process motifs these involve include: the \emph{feedback loops} shown in \fig{motifsForContinuous}, \emph{feedforward loops} of the same length $u$ \subref{fig:k_i_2_k_i_2},
	and general open \emph{dual walk motifs} of the same length $u$ where $j \neq i$ \subref{fig:k_i_1_k_j_1} and \subref{fig:k_i_2_k_j_2}.
	Subfigure labels indicate the weighted counts to which the displayed process motifs contribute.
	Nodes are coloured as per $k$, $i$ and $j$ for $\walkprod{k}{i}{u}{k}{j}{u}$ in \subref{fig:k_i_u_k_j_u}, with cross-hatching on nodes indicating node $j = i$ in \subref{fig:k_i_u_k_i_u} and \subref{fig:k_i_2_k_i_2}.}
\end{figure}

\section{Further analysis and discussion}
\label{sec:discussion}

\subsection{Relative prevalence of closed dual walk process motifs determines synchronizability}
\label{sec:interpretation}

Whilst the precise process motifs involved in determining synchronizability for continuous- and discrete-time processes are subtly different in \eq{synchronizabilityMotifsContinuous} and \eq{synchronizabilityMotifsDiscrete}, the overall story is quite similar.
Synchronizability is reduced (with increasing $\left\langle \sigma^2 \right\rangle$) as the weighted count of closed dual walk motifs increases compared to what may be \textit{expected} from all dual walk motifs ending at any two nodes (not necessarily ending at the same node).
We can phrase this in several ways. 
For node pairs $k$ and $i$ connected by a walk of length $u$, \eq{synchronizabilityMotifsContinuous} can be interpreted as estimating the difference between the actual amount of (weighted) redundancy from $k$ to $i$ of length $m-u$ (captured in $\walkprod{k}{i}{u}{k}{i}{m-u}$) versus the amount expected if the walks of length $m-u$ from $k$ to $i$ were weighted as walks from $k$ to all other nodes $j$ on average (captured in $\frac{1}{N}\sum_j{\walkprod{k}{i}{u}{k}{j}{m-u}}$). The more such redundancy in our walks, the less synchronizable the network becomes (with higher $\left\langle \sigma^2 \right\rangle$).
Put another way, it could be seen as a weighted consideration of whether walks of length $m-u$ from node $k$ -- which has a walk to $i$ in $u$ steps -- preferentially end back at node $i$, as compared to such walks ending at another node $j$ of the network.

Indeed, this is much easier to interpret if we assume for a moment that all edge weights are positive and of constant magnitude, and compare networks of fixed degree.
In this case, the interpretation becomes largely a numbers game, and we see immediately that synchronizability is reduced ($\left\langle \sigma^2 \right\rangle$ is increased) as more walks from nodes $k$ to $i$ become closed by other directed walks from $k$ to $i$ (rather than those walks reaching some other node $j$).
Here then, the more feedback and feedforward process motifs or loops we have in our network (for discrete-time only closed process motifs of the same length), or in other words the more clustered structure we have, the worse the synchronizability of the network becomes.

\subsection{Detrimental effects of clustered structure}
\label{sec:clusteredStructure}

\begin{figure*}
		\sbox0{\subfigure[$d=2$,$c=0.5$]{\label{fig:d2_c0_5}\includegraphics[width=0.95\columnwidth]{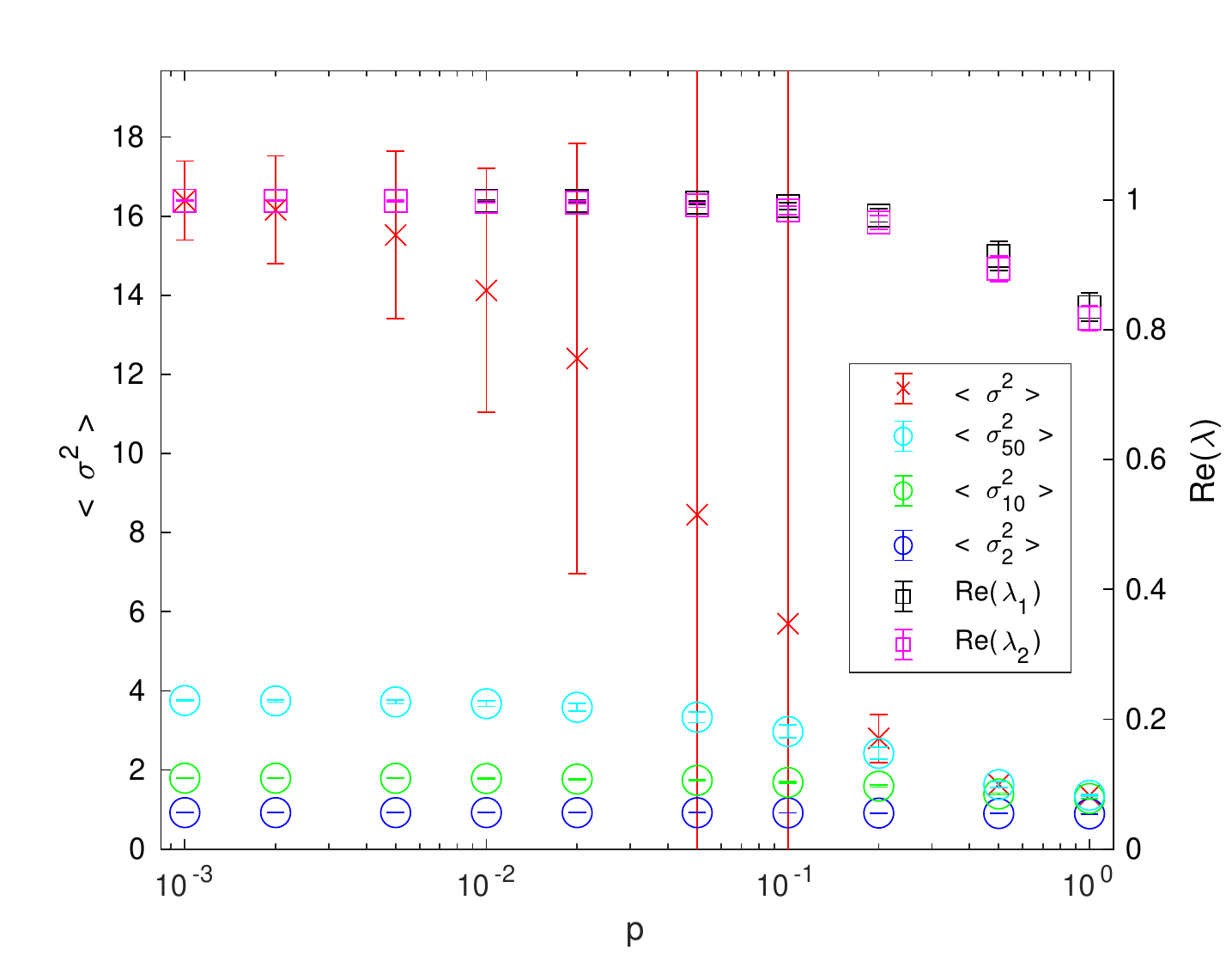}}}
		\sbox1{\subfigure[$d=4$,$c=0.5$]{\label{fig:d4_c0_5}\includegraphics[width=0.95\columnwidth]{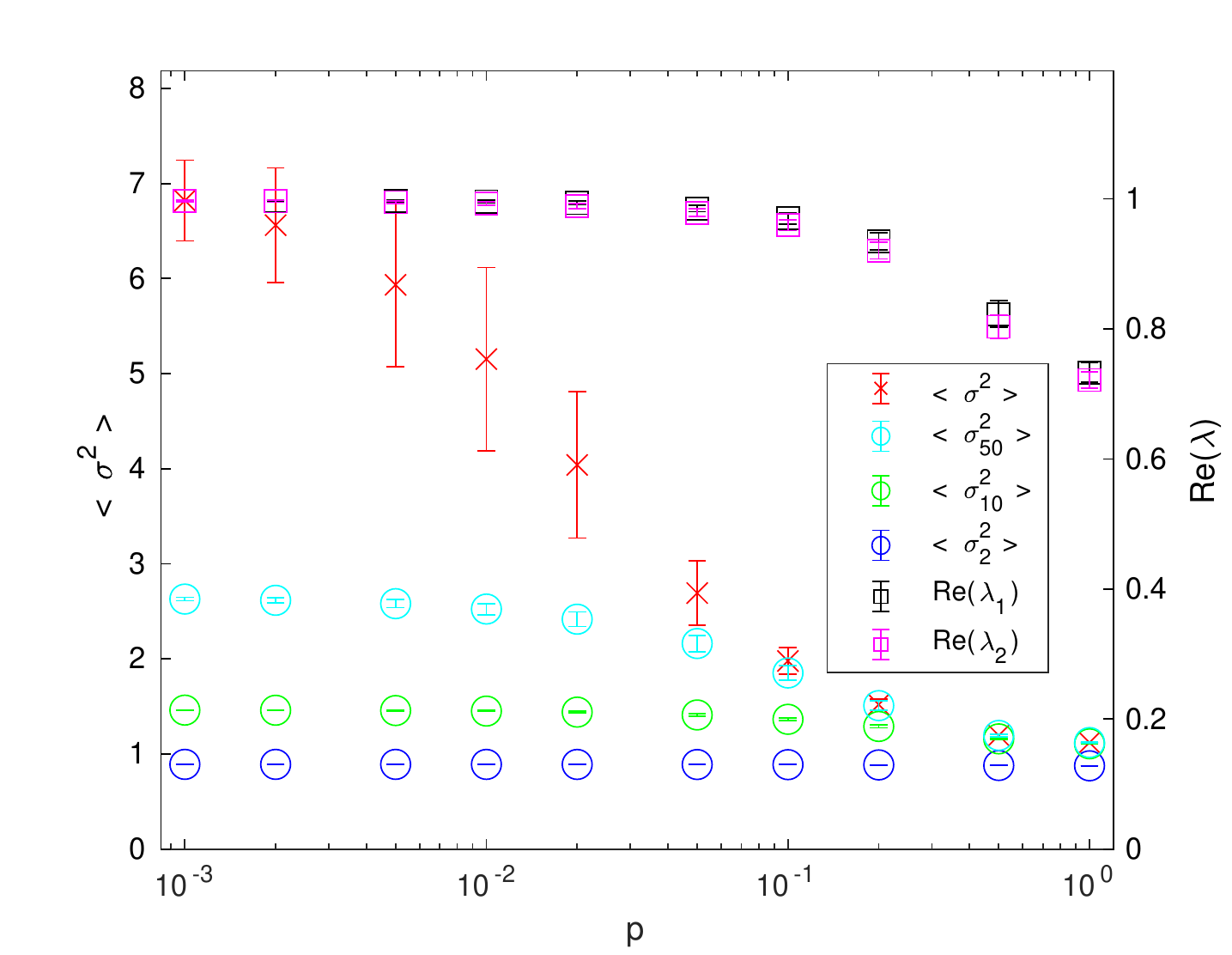}}}
		\sbox2{\subfigure[$d=8$,$c=0.5$]{\label{fig:d8_c0_5}\includegraphics[width=0.95\columnwidth]{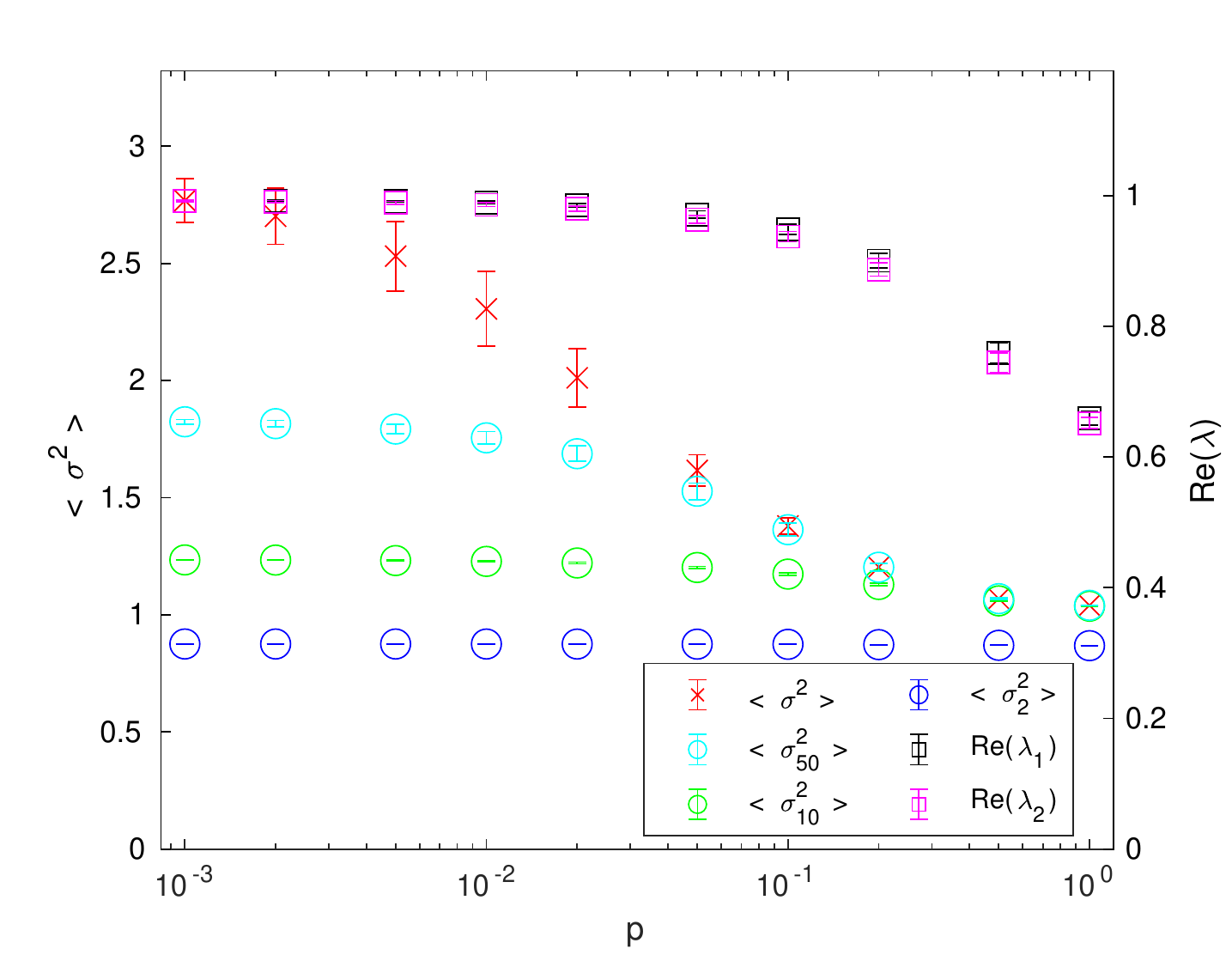}}}
		\sbox3{\subfigure[$d=4$,$c=0.1$]{\label{fig:d4_c0_1}\includegraphics[width=0.95\columnwidth]{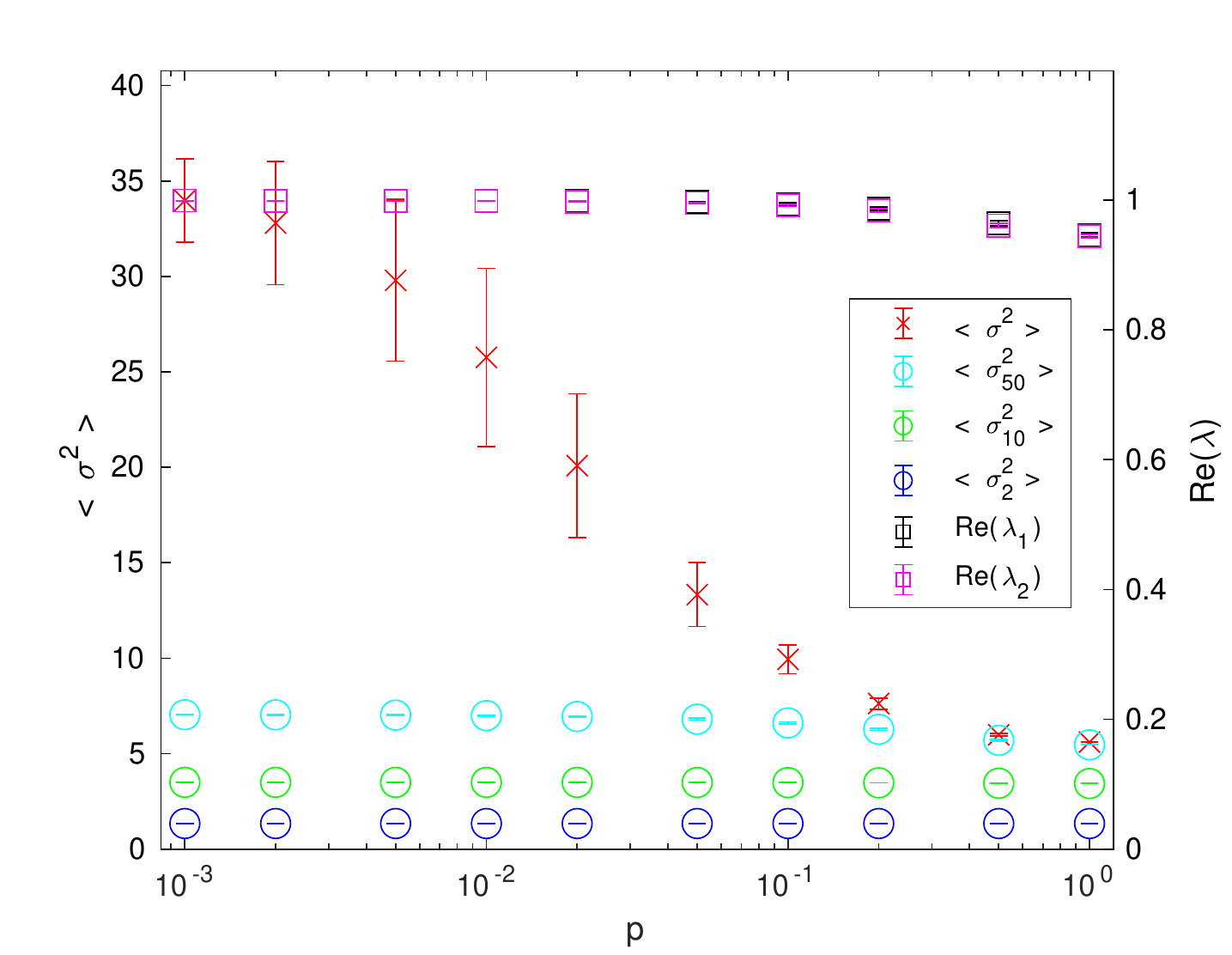}}}
		\sbox4{\subfigure[$d=4$,$c=1.0$]{\label{fig:d4_c1_0}\includegraphics[width=0.95\columnwidth]{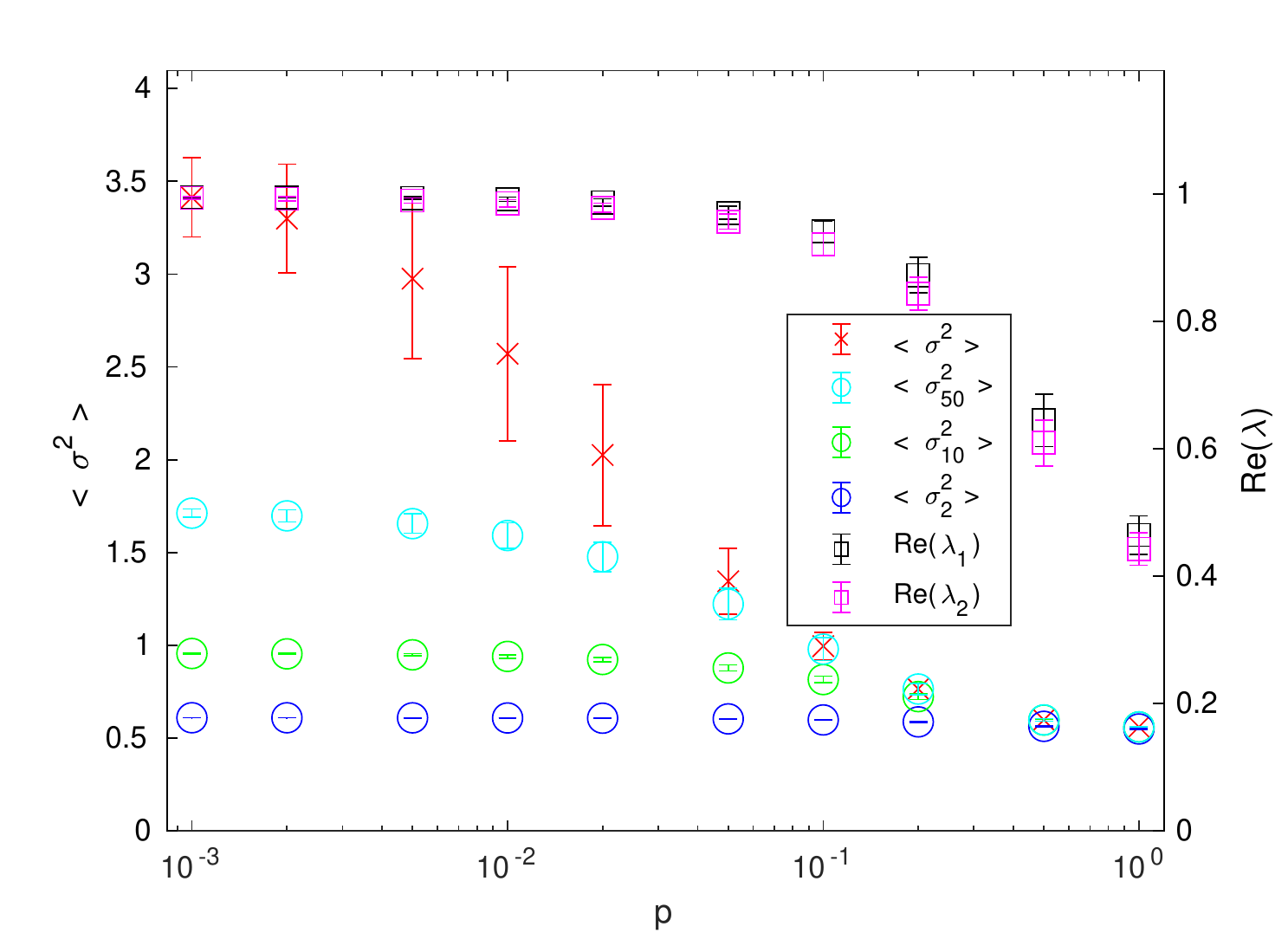}}}
		\sbox5{\subfigure[$d=4$,$p=0.001$]{\label{fig:d4_p0_001}\includegraphics[width=0.95\columnwidth]{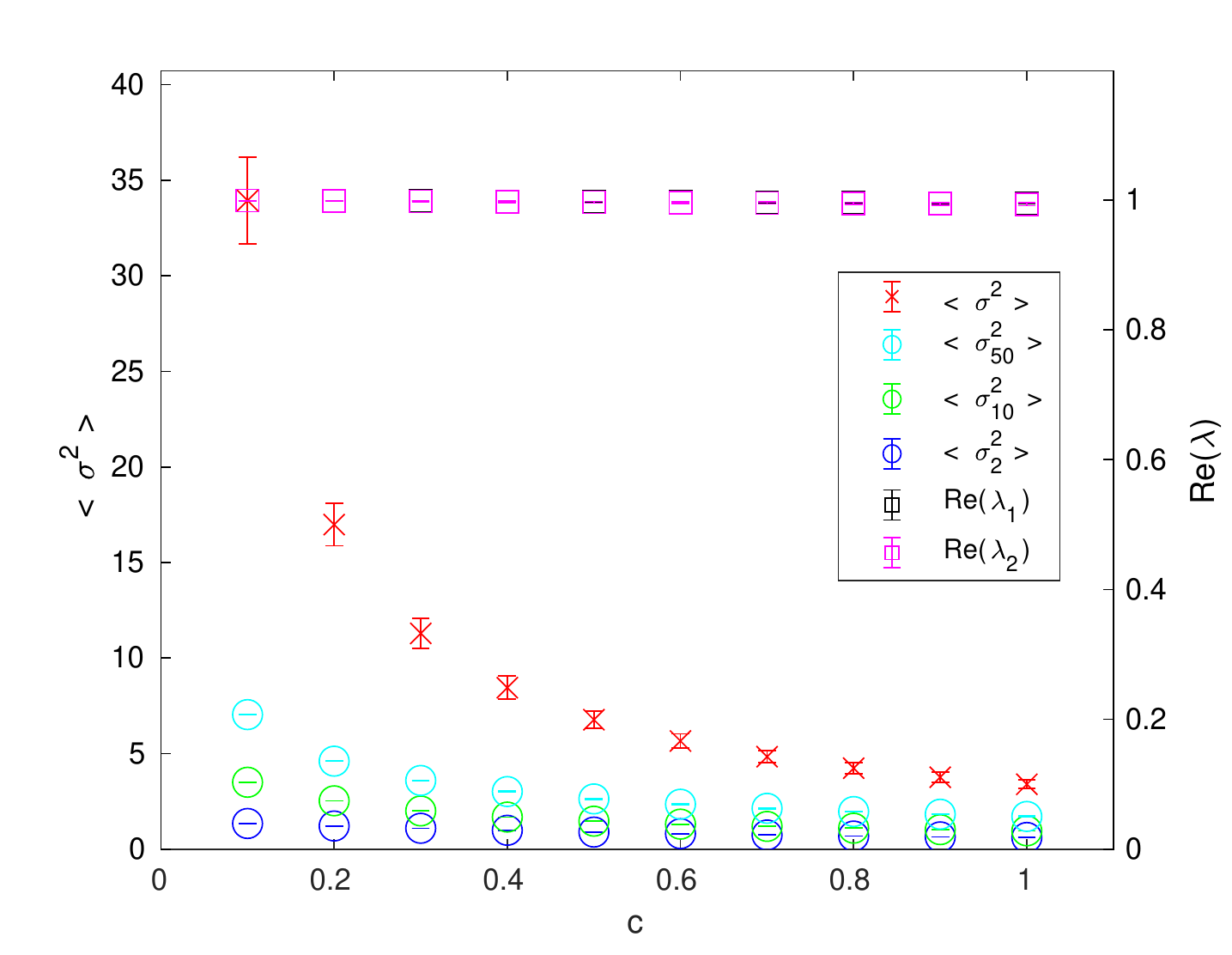}}}
 \centering
  \usebox0\hfil \usebox3\par
  \usebox1\hfil \usebox4\par
  \usebox2\hfil \usebox5
	\caption{\label{fig:smallWorld} \textit{Numerical results for synchronization throughout a small-world transition} for continuous-time dynamics on an $N=100$ Watts-Strogatz ring network.
	Parameters for the network connectivity matrix are described in \fullSubSecRef{Covariance}{numericalValidation}, including in-degree $d$, total incoming (non-self) connection weight $c$, and network randomization parameter $p$ (e.g. in \subrefs{d2_c0_5}{d4_c1_0} regular ring networks appear on the left as $p \rightarrow 0$ and random networks on the right as $p \rightarrow 1$).	
	Various measures and heuristics for the quality of synchronisation computed analytically from the network structure alone (without simulating dynamics) are plotted, including: the average steady-state distance from synchronisation $\left\langle \sigma^2 \right\rangle$ computed via our new method in \eq{synchronizabilityContinuous} (presented in \fullSubSecRef{Covariance}{clusteredStructure}),
	 and low-order approximations $\left\langle \sigma^2_M \right\rangle$ (defined in \eq{synchronizabilityMotifsContinuousLowOrder} in \fullSubSecRef{discussion}{strongerShorter}),
	 and the largest two real components of the eigenvalues of $C$ (except for $\psi_0$, discussed in \fullSubSecRef{discussion}{relation}).
	 The mean value for each measure for each parameter combination is taken over 2000 network realisations and error bars represent the standard deviation.
	 \fig{d2_c0_5_fullErrorBars} in \app{suppFigs} reproduces \subref{fig:d4_c0_5} showing the full extent of the error bars.
}
\end{figure*}

The above provides a direct \textit{explanation} for the result that maximally synchronizable networks have no directed loops \cite{nishikawa2006}, and moreover the widely-observed result \cite{nishikawa2003,atay04a,korniss2003,barahona2002,grabow2010,kim13a,kim2015,aren08} that random networks are more synchronizable than small-world or regular ring networks.
In other words, because small-world or regular networks simply have more feedback and feed-forward process motifs in their more clustered structure, they exhibit worse synchronisability in general.
This result applies directly where edge weights are positive, but can also be expected to be similar when the edge weights are positive on average, as is commonly the case in biological networks such as in the mammalian cortex \cite{barn11a}.

We numerically explore this idea in \fig{smallWorld}, continuing our experiment on sample network structures throughout a small-world transition on an $N = 100$ Watts-Strogatz ring network, in continuous-time with $\theta = \zeta = 1$, with various in-degrees $d$ and total incoming connection weights $c$ as outlined in \fullSubSecRef{Covariance}{numericalValidation}.
Again, Matlab code for these experiments is distributed in the \textit{linsync} toolbox \cite{linsync}.
\subfigs{d2_c0_5}{d4_c1_0}, for a range of different combinations of in-degree $d$ and connection weight $c$, each directly show a strong transition in the steady-state distance from synchronization $\left\langle \sigma^2 \right\rangle$, changing from being very large (i.e. relatively poor synchronization) for regular ring network structures (low $p$) through to a substantial proportion smaller (i.e. relatively good synchronization) for random network structures (high $p$).
The relative synchronizability improves rapidly through the small-world regime; e.g. for $d=4$, $c=0.5$ (\fig{d4_c0_5}) we observe $\left\langle \sigma^2 \right\rangle$ dropping by around 25\% with only 1\% of links randomized, and by a factor of 3 by the time 10\% of links are randomized.

We also see much less variance in $\left\langle \sigma^2 \right\rangle$ across network samples once 5-10\% of the edges are randomized, since the networks then become more homogeneous across nodes and therefore across network samples as well.
(The initial large increase in variance for $d=2$,$c=0.5$ in \fig{d2_c0_5} is discussed in \app{formForSymmetricC}).
This greater homogeneity across nodes also leads to lower variance in dynamics, which would explain the smaller relative error of empirical results against the analytic values of $\left\langle \sigma^2 \right\rangle$ in \fig{convergence}.

The result that regular and small-world networks exhibit worse synchronization in general is subtle and is not just a numbers game though -- we emphasise it is about the weighted \textit{difference} of closed versus all dual walk process motifs.
More clustered structure will lead to worse synchronizability under the above constraints of positive constant magnitude edge weights and fixed degree, but does not necessarily mean worse synchronizability when relevant aspects of the network are not well controlled, as is made clear by Arenas et al. \cite{aren08} and Nishikawa and Motter \cite{nish10a}.
For example, increasing the degree of nodes in a small-world network will lead to higher counts of loop process motifs, but eventually this leads to a fully connected network which is known to have optimal synchronizability \cite{korn07}.
\footnote{We can already see in \subfigs{d2_c0_5}{d8_c0_5} that increasing the in-degree substantially drops $\left\langle \sigma^2 \right\rangle$, to a point that for $d=8$ with $p \rightarrow 0$ it indicates a lower distance from synchronisation than for substantially more randomised networks with $d=2$. This would be confusing if one thought only about closed walk process motif counts, without thinking about the difference between closed and open and properly controlling relevant aspects of the networks.}
As above, one needs to consider not only numbers of such process motifs, but the relative difference between closed and all such counts: in the fully connected network case there may be many closed dual walk process motifs, but there are many more that remain unclosed.
Indeed, this provides a direct argument against extrapolating from studies of synchronizability of motifs in isolation \cite{morenovega2004,lodato2007,li2010}, because the relative proportions become very skewed compared to what they may look like when embedded in a network.

\subsection{Interpretation in terms of perturbation reinforcement}
\label{sec:perturbationInterpretation}

We can also interpret main results in terms of resilience to perturbations on node $k$ in \eq{synchronizabilityMotifsContinuous} and \eq{synchronizabilityMotifsDiscrete}. Closed dual walks from $k$ to $i$ (again considering positive edge weights only) lead to these perturbations eventually reinforcing on node $i$, driving it away from synchronising with the other nodes (as observed for variance for the same dynamics by Schwarze and Porter \cite{schwarze2021}).
In contrast, a greater proportion of walks to other nodes $j$ will dissipate such perturbations more evenly through the network, and thereby disturb synchronization to a lesser degree.

We can also look at this in terms of competition between clusters for what the dominant value to synchronise around will be: more clustering leads to larger perturbation reinforcement within the cluster and stronger competition between the clusters, which harms overall global synchronisation.
Indeed, this illustrates the known result that increasing synchronisabilty within two largely separate network clusters can harm synchronisabilty of the overall network \cite{atay05a}.
From an information-theoretic perspective, loop process motif structures increase information transfer within such clustered modules \cite{novelli2020}, reinforcing synchronisation within clusters at the expense of synchronisabilty of the overall network.

\subsection{Contrasting continuous- and discrete-time results}
\label{sec:contDiscTime}

Indeed, this interpretation in terms of perturbations allows us to see why the contributing motifs are different for continuous- and discrete-time processes.
For discrete-time processes, perturbations on node $k$ will only reinforce on node $i$ from two walks of the \textit{same} length, since the driving noise $\vec{r}(t)$ is uncorrelated in time.
In continuous-time processes, considered to be more realistic for modelling biological processes \cite{barn09b}, nodes in contrast experience an uncorrelated rate $d\vec{w}(t)$ of perturbations, though their integration over time means that perturbations on node $k$ can be reinforced on node $i$ from two walks of different lengths. Hence, more general closed dual walk motifs contribute to synchronizability for continuous-time processes, including feedforward process motifs with \textit{different} length walks.
As we see from the $\binom{m}{u}$ term in \eq{synchronizabilityMotifsContinuous}, the more similar the lengths of the dual walks of length $u$ and $m-u$ are, the stronger their contribution (for fixed $m$). This was also observed for covariance for the same dynamics \cite{schwarze2021}, since this ensures a maximal reinforcement of the perturbation.

\subsection{Stronger contribution of shorter walks}
\label{sec:strongerShorter}

Shorter walks contribute more strongly to the expressions for $\left\langle \sigma^2 \right\rangle$ for these dynamics, as has been observed for other measures \cite{barn09b,lizier2012storageLoops,schwarze2021}, since under the stability of the dynamics (orthogonal to $\psi_0$) information is gradually ``forgotten'' along the relevant walks.

We can investigate this numerically by defining low-order approximations $\left\langle \sigma^2_M \right\rangle$ to $\left\langle \sigma^2 \right\rangle$, for motif sizes up to an order $M$. In continuous-time, $M$ represents the longest total length of the two dual walks $(m-u, u)$ from $k$ to $i$ considered:
\begin{align}
	\left\langle \sigma^2_M \right\rangle =  
			\frac{\zeta^2}{2\theta} \sum_{m=0}^{M} & \frac{2^{-m}}{N} \sum_{u=0}^{m} \binom{m}{u} \times \nonumber \\
				& \sum_{i,k}{
			\left( \walkprod{k}{i}{u}{k}{i}{m-u} - \frac{1}{N} \sum_{j}{ \walkprod{k}{i}{u}{k}{j}{m - u} } \right) },
	\label{eq:synchronizabilityMotifsContinuousLowOrder}
\end{align}
whilst in discrete-time, $M$ is more naturally defined as the longest length of each of the dual walks of length $u$ from $k$ to $i$ that is considered:
\begin{align}
	\left\langle \sigma^2_M \right\rangle & = \frac{\zeta^2}{N} \sum_{u=0}^{M}{\left( \sum_{i,k}{ \walkprod{k}{i}{u}{k}{i}{u}	} - \frac{1}{N}\sum_{i,j,k}{ \walkprod{k}{i}{u}{k}{j}{u}	} \right)}.
	\label{eq:synchronizabilityMotifsDiscreteLowOrder}
\end{align}

These low-order approximations $\left\langle \sigma^2_M \right\rangle$ for $M = \{2,10,50\}$ are plotted for the various numerical experiments in \fig{smallWorld}, and in all cases we see that the the low-order approximations for $M = \{2,10\}$ account for a proportionally larger component of $\left\langle \sigma^2 \right\rangle$ than the higher-order contributions that would also be included in $\left\langle \sigma^2_{50} \right\rangle$. This directly shows that the shorter walks contribute more strongly than the longer walks.
This effect is particularly striking for more random networks ($p \rightarrow 1$), where the low order approximations are already very close to the true value of $\left\langle \sigma^2 \right\rangle$.
This is because the non-trivial amount of clustered structure in the more regular networks induces many more convergent walks, particularly at longer walk lengths, which need to be accounted for.
The effect is less clear for low cross-coupling cases such as $c=0.1$ in \fig{d4_c0_1} where walks involving the more heavily weighted self-links ($C_{ii}=1-c=0.9$) dominate, and since these are not randomised then there is less difference in the weighted prevalence of convergent walks across network types at shorter walk lengths.

As the coupling strengths are increased though, the longer walks begin to have larger impacts on $\left\langle \sigma^2 \right\rangle$.
This effect was explored for other measures by \cite{barn09b,barn11a}.
We can see directly from \eq{synchronizabilityMotifsContinuous} that network-wide increases in coupling will non-linearly increase $\left\langle \sigma^2 \right\rangle$, until reaching a point where $\left\langle \sigma^2 \right\rangle$ diverges (meaning the network is not synchronizable).
This also reflects the result that any network structure can be made synchronizable by lowering the magnitudes of the $C_{ji}$ \cite{atay06a}.
Illustrating the increased impact of longer walks with coupling strength is somewhat complicated in the numerical experiments here though, because increasing the cross-coupling $c$ necessarily means reducing the self-coupling $C_{ii}=1-c$ (required to maintain $\psi_0$), and both contribute to $\left\langle \sigma^2 \right\rangle$.
Indeed, \fig{d4_p0_001} shows that $\left\langle \sigma^2 \right\rangle$ actually \textit{decreases} as the cross-coupling $c$ increases as observed by \cite{korn07}
(although the relative proportion accounted for by the longer walks increases) -- but again this is solely because of the related increase in self-coupling $C_{ii}=1-c$ in our model, which dominates the changes to $\left\langle \sigma^2 \right\rangle$ there.
A more clear illustration emerges if we vary in-degree $d = \{2,4,8\}$ for fixed total incoming weights $c=0.5$ and therefore fixed self-coupling $C_{ii}=1-c=0.5$ in \subfigs{d2_c0_5}{d8_c0_5}: lowering the in-degree increases the weights on each non-self edge, and although this means there are less cross-edges and less shorter walks, their larger weights leads to a substantial increase in $\left\langle \sigma^2 \right\rangle$ away from synchronisation, and a larger influence of the longer walks in that increase.

\subsection{Lower synchronizability for symmetric networks}
\label{sec:symmetricWorse}

Aside from direct connections contributing via $\walkprod{k}{i}{1}{k}{i}{1}$, the shortest and therefore strongest contributing closed dual walk process motifs are recurrent connections $\walkprod{i}{i}{1}{i}{i}{1}$ and feedback loops of length 2 which are involved e.g. in $\walkprod{i}{i}{2}{i}{i}{u}$.

The latter are of particular interest, since in undirected or symmetric networks $C$ reciprocal connections -- feedback loops of length 2 -- occur for \textit{every} connected pair of nodes.
In comparison to directed networks with the same fixed in-degree, and fixed positive edge weights (or majority positive as above), this will make substantial contributions to increasing $\left\langle \sigma^2 \right\rangle$ in \eq{synchronizabilityMotifsContinuous} and \eq{synchronizabilityMotifsDiscrete}, and thereby reducing synchronizability.
In other words, this explains why undirected networks are generally less synchronizable than corresponding directed networks (as observed e.g. in \cite{motter2005}).

Considering the $\walkprod{i}{i}{2}{i}{i}{2}$ as one such term acting to increase $\left\langle \sigma^2 \right\rangle$ for example, it is easily demonstrated for fixed positive edge weights that the only way for that term to be completely countered by the corresponding $\frac{1}{N} \sum_{j}{ \walkprod{i}{i}{2}{i}{j}{2} }$ term would be for each outgoing neighbour of $i$ to connect onwards to every node in the network. (Which of course, reflects the aforementioned situation in full connectivity where closed dual walk process motifs lose any sense of being comparatively common).

Another perspective on this result is that symmetric connectivity is known to increase autocorrelation and slow down dynamics \cite{marti2018}.
From an information-theoretic perspective, symmetric connectivity and loop process motif structures in general are known to increase self-predictability of nodes embedded in a network, as measured by active information storage \cite{lizier2012storageLoops}, and higher active information storage has been observed to be negatively correlated to synchronizability \cite{ceg11}.
Such self-predictability may indicate a node is less responsive to the dynamics of others in the network, and therefore not synchronizing as well with them.

\subsection{Inhibition versus excitation}
\label{sec:inhibition}

Much of the above intution comes from considering positive of excitatory edge weights only, or majority positive edge weights, and is generally flagged as such.
We must however be aware that the presence of negative or inhibitory edge weights can have dramatic effects on our interpretation of the results here. The mathematics of \eq{synchronizabilityMotifsContinuous} and \eq{synchronizabilityMotifsDiscrete} are unchanged, we simply need to remember that the weighted walk counts incorporate these negative edge weights.
This has been emphasised particularly by Nishikawa and Motter \cite{nish10a} and is very relevant for this study.

As an example, consider the reciprocally connected pairs referred to within the term $\walkprod{i}{i}{2}{i}{i}{u}$ above as reducing synchronizability -- for positive edge weights.
Where one edge weight in the pair is negative, this switches the sign of the contribution and therefore leads to the relevant term actually \textit{enhancing} synchronization -- assuming that all of the other weights in the other feedback loop $\walkmotif{i}{i}{u}$ as part of this term are positive.
This illustrates the difficulty in extending our intuition to the presence of negative edge weights, since multiple instances of them can be quite confusing.

\subsection{Relation to previous approaches}
\label{sec:relation}

We emphasise the generality of our result in comparison to previous work in not requiring symmetry, nonnegativity or diagonalizability of the weighted coupling matrix $C$.
This is important since real networks are often directional and weighted, and since optimal synchronization requires nondiagonalizable networks \cite{nishikawa2006}.
Indeed, non-diagonalizable matrices $C$ are encountered even with our very simple model for generating Watts-Strogatz ring networks for the numerical experiments in \fig{smallWorld}, e.g. being a substantial proportion of networks for $d=2$, $c=0.5$ from $p \geq 0.05$ and becoming the majority for larger $p$.

Whilst extremal eigenvalues determine whether sync is possible at all, as they control the fastest and the slowest decaying modes, \eq{symmetricCcontinuous} here emphasises that fully characterising the relative synchronizability requires taking \textit{all} of the eigenvalues of $C$ into account (as echoed by \cite{arenas2006}).
This contrasts with common heuristics for relative synchronizability which focus on the extremal eigenvalues or their spread (as per \fullSubSecRef{measuringSync}{syncConditions}).
The common heuristic for continuous-time dynamics is that synchronization is improved as  $\mathrm{Re}(\lambda_1)$, the largest real component of an eigenvalue of $C$ (aside from $\lambda_0=1$ corresponding to $\psi_0$) moves further away from the stability boundary of 1; or in other words synchronization improves as $\mathrm{Re}(\lambda_1)$ decreases.\footnote{This is complementary to the expression of synchronization improving as $\mathrm{Re}(\lambda_{L1})$ for the smallest eigenvalue of $L$ increases \cite{atay05a}.}
To investigate how well this heuristic performs, $\mathrm{Re}(\lambda_1)$ is plotted in the results of our numerical experiments in \fig{smallWorld}.
Each sub-plot -- across a wide range of parameters -- shows that this heuristic provides a remarkably different indication of how the quality of synchronization changes with respect to $p$ and $c$, as compared to the ground truth now provided by $\left\langle \sigma^2 \right\rangle$. In particular, we see in \subfigs{d2_c0_5}{d4_c1_0} that whilst changes in $\left\langle \sigma^2 \right\rangle$ towards improved synchronization are visible well before network randomization parameter $p$ reaches $0.01$ (or on average changes to 4 out of 400 edges here), the heuristic of $\mathrm{Re}(\lambda_1)$ requires an order of magnitude   larger randomisation ($p$ around 0.1) before any change is visible.
This contrast is more stark in the experiment varying cross-coupling strength in \fig{d4_p0_001}, where there is no visible change at all in the heuristic $\mathrm{Re}(\lambda_1)$, despite the underlying distance from synchronization $\left\langle \sigma^2 \right\rangle$ indicating an order of magnitude improvement.
The second-largest real component $\mathrm{Re}(\lambda_2)$ is also included in these plots, with similar results.
This result underlines the importance of using our precise measure of synchronizability rather than the common heuristics from the extremal eigenvalues.

There are other interesting mathematical connections from our result to previous work.
\app{formForSymmetricC}, culminating in \eq{symmetricCcontinuous}, shows for instance how the result for continuous time in \eq{synchronizabilityMotifsContinuous} simplifies for symmetric networks with $\theta=\zeta=1$ to give the existing result \eq{widthAsInverseEigs} by Korniss \cite{korn07} for these in terms of inverse sums of eigenvalues, as it should.
We also contribute a similar simplification for symmetric networks in discrete-time with $\zeta=1$ (derived in \app{formForSymmetricC}):
\begin{align}
	\left\langle \sigma^2 \right\rangle = \frac{1}{N} \sum_{\lambda \neq \lambda_0} \frac{1}{ 1 - \lambda_C^2 }.
	\label{eq:symmetricCdiscreteMainText} 
\end{align}
Whilst the two continuous-time results are equivalent when $|\lambda_v| < 1$ (except for $\lambda_0$), our power series result does not converge outside this range (within $\mathrm{Re}(\lambda_1) < 1$) even though $\left\langle \sigma^2 \right\rangle$ has a solution for symmetric networks from \eq{widthAsInverseEigs}; future work may establish how our solution could be adapted to provide further insight in these conditions.

Further insights using \eq{widthAsInverseEigs} for symmetric networks with $\theta=\zeta=1$ are found by interpreting $C$ as the transition matrix for random walks on a network (with $C_{ji} \geq 0$ being probabilities when $\vec{\psi_0}$ is an eigenvector with eigenvalue $\lambda_0=1$).
Here then, we can immediately observe that the steady-state distance from synchronisation $\left\langle \sigma^2 \right\rangle$ for $C$ in continuous-time in \eq{widthAsInverseEigs} is directly related to Kemeny's constant $K(C)$ \cite{levene2002,riascos2021}, as $\left\langle \sigma^2 \right\rangle = \frac{K(C)}{2N}$.
Kemeny's constant captures the (weighted) average travel times of random walks between two nodes in the network, and can also be interpreted in terms of how well-connected the nodes are, and as an indication of clustering \cite{breen2022}.
Kemeny's constant also directly maps to the ``effective graph resistance'' or Kirchhoff index for regular graphs \cite{palacios2010,wang2017}, which is interpreted as the difficulty of transport in the network, and so we see that $\left\langle \sigma^2 \right\rangle$ also directly maps to this quantity in this case.
These interpretations are consistent with the insights our new power series formalism brings for $\left\langle \sigma^2 \right\rangle$, e.g. the relations to clustering.
Moreover, the insights our formalism brings in terms of the relative abundance of convergent process motifs are thus directly applicable for Kemeny's constant and effective graph resistance, and will find further utility in those studies.
This is particularly important for Kemeny's constant for example, where gaining intuition into the measure has been an ongoing concern \cite{levene2002,breen2022}.
Similarly, our process motif interpretation could be usefully applied to the use of effective graph resistance to study the robustness of complex networks, which has known applications in mitigating cascading failures in electrical networks \cite{wang2017}.
Finally, we note that most work on Kemeny's constant has considered undirected networks \cite{breen2022,riascos2021,wang2017}, and as such our directed process motif interpretation may contribute to studies on directed networks \cite{kirkland2016,pitman2018} here.

Another interesting connection is found in combining \eq{covarianceInProjected} and \eq{traceInOrthogonalBookend}, giving:
\begin{align}
	\left\langle \sigma^2 \right\rangle & = \frac{1}{N} \sum_i{\Omega_{ii}} - \frac{1}{N^2} \sum_{i,j}{ \Omega_{ij} }
	\label{eq:syncWithCovariances}
\end{align}
This directly shows us that $\left\langle \sigma^2 \right\rangle$ is a difference between the average autocovariance of all nodes, and the average covariance between all node pairs (including the autocovariances).
Crucially, this difference converges under the stated conditions via the proofs referred to in \fullSubSecRef{Covariance}{powerSeriesMethod}, even though the various (auto)covariance terms may not (i.e. where $\psi_0$ is an eigenvector).

Indeed, this difference maps directly back to the difference between closed and all dual walk motif counts in our results: the weighted counts for dual walk motifs converging on node $i$ determine the autocovariance for node $i$, whilst the weighted counts of dual walk motifs terminating at nodes $i$ and $j$ determine their covariance.
That these types of process motifs determine the autocovariance and covariance was previously observed by Schwarze and Porter \cite{schwarze2021}, and they necessarily recur here in determining synchronizability as demonstrated above.

We also note that the mathematics of the dual walk motif counts in \eq{synchronizabilityMotifsContinuous} and \eq{synchronizabilityMotifsDiscrete} have some resemblence to the power series expansion of the exponential of the adjacency matrix in the communicability function \cite{estrada_communicability_2008} (albeit with different weightings for the contributions of different walk lengths), and with the terms forming the difference in \eq{syncWithCovariances} then bearing resemblence to self- and non-self-communicability themselves \cite{estrada_physics_2012}.

\subsection{Biological implications}
\label{sec:bioImplications}

Whilst our results consider walks on networks, or process motifs, rather than directly counting structural motifs, we know that increasing counts of structural motifs with closed paths will lead to increasing counts of the closed walk motifs observed here \cite{schwarze2021} (specifically where edge weights are positive on average, as observed earlier being commonly the case in biological networks such as in the mammalian cortex \cite{barn11a}).
And indeed, such closed loop structural motifs are observed in many real world networks.
 
A particularly interesting example, where synchronization is of special relevance, lies in brain networks.
The high level of modularity in brain networks -- which is particularly important in compartmentalising function for cognition, and in conserving wiring length \cite{spo11a} -- also serves to increase counts of closed loop motifs, and therefore push the brain away from global synchronization.
It is interesting to note that synapses of bidirectionally connected pairs of neurons are on average stronger than synapses of unidirectionally connected pairs \cite{marti2018}, which will also specifically act against global synchronization.
These results are important since synchronization is associated with pathological conditions such as epilepsy and Parkinson's disease \cite{aren08}.
Regional synchronization within a module in contrast may be useful, and could be enhanced by more local connectivity within modules.

\section{Conclusion}
\label{sec:conclusion}

In this paper, we have provided the first fully general analytic calculation of synchronizability which requires neither symmetric nor diagonalizable connectivity matrix $C$.
This is an important theoretical advance in and of itself, allowing us to handle network structures which are more commonly found and to compute the steady-state distance from synchronisation $\left\langle \sigma^2 \right\rangle$ without such restrictions for the first time.
Furthermore, this provides a substantial advance over heuristic approximations from the eigenvalue spectra (e.g. examining leading eigenvalues and the spread of the spectrum) that had been previously in use, allowing us to properly compute the synchronizability for a given network structure here.

Perhaps more importantly, we are able to interpret our new result in terms of weighted directed walks on networks, or process motifs.
Indeed, this allows us to fully comprehend the role of network structure in determining synchronizability.
Along these lines, we can now completely explain several well-known results for the first time, including that regular ring and small-world networks tend to be further from synchronization than random networks, and the manner in which clustered structure leads to reduced global synchronization.

These results are more important more generally in characterizing the relationship between network structure and dynamics or function, since addressing this for synchronization has historically been seen as one of the most important canonical problems in this field \cite{aren08,porter2016}.

Along these lines, the results here will have applicability beyond the apparent scope of the Ornstein-Uhlenbeck or Edwards-Wilkinson process and vector-autoregressive dynamics, with the former having been observed as
``entirely equivalent'' to a linearized Wilson-Cowan noisy neural network \cite{barn09b}, and is also similar to covariances of spike rates in systems of coupled Hawkes processes \cite{pernice2011,hu2014}.

There is much scope for further extension of the results included here.
Primarily, the exact results for the steady-state distance from synchronization provided here should be used to systematically investigate the plethora of insights that have been gained in the past using heuristics from the eigenspectra -- as we have shown here such heuristics can easily miss very drastic changes in the synchronization characteristics as we alter network properties.

\newcommand{\ackText}{JL was supported through the Australian Research Council DECRA grant DE160100630,
and The University of Sydney SOAR award.
FB was partially supported by the Alexander von Humboldt foundation and the NSF Grant DMS-0804454 Differential Equations in Geometry.
The authors wish to thank Lionel Barnett and Gyorgy Korniss for clarifying details of their work, and Alice Schwarze, Mac Shine, Ben Fulcher and Karel Devriendt for comments.}

\begin{acknowledgments}
\ackText
\end{acknowledgments}

\bibliography{refs}

\newpage

\renewcommand\appendixname{Supplementary Information}

\appendix
\section{Properties of the centering projection}
\label{app:traceInOrthogonalSpace}

First we define the \emph{centering matrix} \cite{mard95}, the projection $U = I - G$ where $G$ is the \emph{averaging operator} $G_{ij} = 1/N$ (for  $N \times N$ matrices).
The operator $U$ removes the average of the given vector it acts on:
\begin{align}
	\vec{x}U = \vec{x} - \bar{x}\vec{\psi_0},
\end{align}
such that for any vector $\vec{x}$, the average of $\vec{x}U$ is $0$.
This serves to project $\vec{x}$ onto the space \textit{orthogonal} to the synchronized state vector $\vec{\psi_0} = [ 1, 1, \ldots , 1 ]$.
Clearly we have $U^T = U$ and $GU = UG = 0$.
In the next lemmata we list a number of useful properties of $U$:

\begin{lemma}
\label{lem:psi0U}
\begin{equation}\vec{\psi_0} U = 0. \label{eq:psi0U} \end{equation}
\end{lemma}
\begin{proof}$\vec{\psi_0} U = \vec{\psi_0} - \bar{\vec{\psi_0}}\vec{\psi_0} = 0$.
\end{proof}

\begin{lemma}
\label{lem:vUv}
For $\vec{v} \perp \vec{\psi_0}$, we have:
\begin{equation} \vec{v} U = \vec{v}. \label{eq:vUv}\end{equation}
\end{lemma}
\begin{proof}Since $\vec{v} \perp \vec{\psi_0}$, we have $\vec{v} \cdot \vec{\psi_0} =  0$ which implies:
$$\frac{1}{N} \sum_j v_j = 0.$$ This means nothing but:
$$\vec{v} G =  0, $$ and therefore: $$\vec{v} U =  \vec{v} (I-G) = \vec{v}.$$
\end{proof}

\begin{lemma}For a matrix $C$ with eigenvector $\vec{\psi_0}$ we have the following properties:
\begin{align}
U C U & = C U, \label{eq:CUequalsUCU} \\
U C^T U & = U C^T, \label{eq:UCTequalsUCTU} \\
C^m U & = (C U)^m, \; m \in \mathbb{N}, m \geq 1 \label{eq:CmUEqualsCUm} \\
U (C^m)^T & = ((C U)^m)^T, \; m \in \mathbb{N}, m \geq 1. \label{eq:UCmTEqualsCUmT}
\end{align}
\end{lemma}
\begin{proof}
Since $G$ is composed of rows $1/N\vec{\psi_0}$, and $\vec{\psi_0}$ is an eigenvector of $C$ with eigenvalue say $\lambda_0$ it is easy to see that: $$GC = \lambda_0G.$$ Now we have: $$UCU  =  (C - GC)U = (C-\lambda_0 G)U = CU,$$ where we used that $GU = 0$.
This proves \eq{CUequalsUCU}, while \eq{UCTequalsUCTU} follows in a similar way by using $UG=0$ or directly by transposing \eq{CUequalsUCU}.

Next, we observe that: $$C^m U = C^{m-1} C U = C^{m-1} U C U,$$ where we used \eq{CUequalsUCU} in the last equality. Iterating this process $m-1$ times in total directly gives \eq{CmUEqualsCUm}. Again \eq{UCmTEqualsCUmT} follows in a similar way, or by directly transposing \eq{CmUEqualsCUm}. 
\end{proof}

\begin{lemma}
\label{lem:vU0}
If $\vec{v}U = 0$, then $\exists d, \vec{v} = d \vec{\psi_0}$.
\end{lemma}
\begin{proof}
If $\vec{v}U = 0$:
\begin{align}
	\vec{v}U & = \vec{v}(I - G) = 0, \\
	\therefore \vec{v} & = \vec{v}G, \\
	\therefore \forall i: \ v_i & = \frac{1}{N} \sum_i v_i , \\
	\mathrm{i.e.} \  \exists d, \ \vec{v} & = d \vec{\psi_0}. \label{eq:vUImpliesPsi0}
\end{align}
\end{proof}

\begin{lemma}
\label{lem:traceInOrthogonalBookend}
The trace of a given matrix $A$ after the centering projections are applied to produce $\mathrm{trace}(UAU)$ is given by:
\begin{equation} \mathrm{trace}(UAU) = \sum_i{A_{ii}} - \frac{1}{N} \sum_{i,j}{ A_{ij} }
\label{eq:traceInOrthogonalBookendAppendix}
\end{equation}
\end{lemma}
\begin{proof}We have:
\begin{align}
 \mathrm{trace}(UAU) =  &\ \mathrm{trace}((I - G)A(I - G)), \nonumber \\
 				= &\ \mathrm{trace}(A - AG - GA + GAG), \nonumber \\
 				= &\ \sum_i{A_{ii}} - \frac{2}{N} \sum_{i,j}{ A_{ij} } \nonumber \\
 				  &\ \ \ + \sum_i{\left( \frac{1}{N^2} \sum_{k,l}{ A_{kl} } \right)}, \nonumber \\
 				= &\ \sum_i{A_{ii}} - \frac{1}{N} \sum_{i,j}{ A_{ij} } \nonumber
\end{align}
\end{proof}

\section{Spectral radius after centering projection}
\label{app:unaveragingZeroMode}

\begin{theorem}
\label{th:specRadiusReducedBelow1}
For a matrix $A$ with eigenvalues $|\lambda_A| < 1$, except for $\lambda_0 = 1$ corresponding to the zero-mode eigenvector $\vec{\psi_0}$, we have $|\lambda_{AU}| < 1$ for all eigenvalues of $A U$ (i.e. $\rho(A U) < 1$).
\end{theorem}
\begin{proof}
By \lemmaRef{eigenvectorsOfUnaveragedProjection} $AU$ only has eigenvalues $\lambda_{AU} = 0$ (including for eigenvector $\vec{\psi_0}$ by \lemmaRef{zeroModeAfterProjection}) or $\lambda_{AU} \neq 0$ for some eigenvectors $\vec{w}$ with $\vec{w} \cdot \vec{\psi_0} = 0$.

Then by \lemmaRef{orthogonalEigenvectorsAfterProjection} each eigenvector $\vec{w}$ of $AU$ with $\vec{w} \cdot \vec{\psi_0} = 0$ has an eigenvalue $\lambda_{AU}$ which is also an eigenvalue of $A$ for some eigenvector $\vec{v} \neq \vec{\psi_0}$, and so by the assumptions of the theorem $|\lambda_{AU}| < 1$.

Therefore, $\rho(A U) < 1$.
\end{proof}

\begin{lemma}
\label{lem:zeroModeAfterProjection}
$\vec{\psi_0} A = \lambda_0 \vec{\psi_0} \rightarrow \vec{\psi_0} AU = 0$.
\end{lemma}
\begin{proof}
\begin{align}
	\vec{\psi_0} A U & = \lambda_0 \vec{\psi_0} U, \\
					& = 0,
\end{align}
via \lemmaRef{psi0U}.
\end{proof}

\begin{lemma}
\label{lem:eigenvectorsOfUnaveragedProjection}
$\vec{\psi_0} A = \lambda_0 \vec{\psi_0}$, $\vec{v} A U = \lambda_v \vec{v}$ $\rightarrow$  $\lambda_v = 0$ or else $\lambda_v \neq 0$ with $\vec{v} \cdot \vec{\psi_0} = 0$.
\end{lemma}
\begin{proof}
Consider a general eigenvector $\vec{v}= a \vec{w} + b\vec{\psi_0}$ of $AU$ with eigenvalue $\lambda_v$ for some $\vec{w} \perp \vec{\psi_0}$, and scalars $a$ and $b$.
By definition $\vec{v}AU = \lambda_v \vec{v}$, so via \eq{CUequalsUCU}:
\begin{align}
\vec{v}UAU = \vec{v}AU = \lambda_v \vec{v}.
\label{eq:eigenOfAU1}
\end{align}
Now, via \lemmaRef{psi0U}:
\begin{align}
\vec{v}UAU & = (a \vec{w} + b\vec{\psi_0})UAU = a \vec{w} AU,
\label{eq:eigenOfAU2}
\end{align}
but also:
\begin{align}
\lambda_v \vec{v} & = a \lambda_v \vec{w} + b \lambda_v \vec{\psi_0},
\label{eq:eigenOfAU3}
\end{align}
so substituting \eq{eigenOfAU2} and \eq{eigenOfAU3} into \eq{eigenOfAU1} we have:
\begin{align}
a \vec{w} AU & = a \lambda_v \vec{w} + b \lambda_v \vec{\psi_0}, \\
(\times U:)\ \  a \vec{w} AUU & = a \lambda_v \vec{w}U + b \lambda_v \vec{\psi_0}U, \\
a \vec{w} AU & = a \lambda_v \vec{w},
\label{eq:eigenOfAU6}
\end{align}
since $U$ is idempotent and via \lemmaRef{psi0U} and \lemmaRef{vUv}.

From \eq{eigenOfAU1} and \eq{eigenOfAU2} we have $a \vec{w} AU = \lambda_v \vec{v}$ and then substituting into \eq{eigenOfAU6} we have:
\begin{align}
\lambda_v \vec{v} & = a \lambda_v \vec{w}.
\end{align}
So either $\lambda_v = 0$, or $\lambda_v \neq 0$ with $\vec{v} \perp \vec{\psi_0}$ (since $\vec{w} \perp \vec{\psi_0}$).
(Note if $a=0$, then $\vec{v}=0$ and we still have $\vec{v} \cdot \vec{\psi_0}= 0$.)
\end{proof}

\begin{lemma}
\label{lem:orthogonalEigenvectorsAfterProjection}
For a matrix $A$ with eigenvalues $|\lambda_A| < 1$, except for $\lambda_0 = 1$ corresponding to the zero-mode eigenvector $\vec{\psi_0}$, if $\vec{w}$ with $\vec{w} \cdot \vec{\psi_0} = 0$ is an eigenvector of $AU$ with eigenvalue $\lambda_w$, then $\lambda_w$ is also an eigenvalue of $A$ for some eigenvector $\vec{v} \neq \vec{\psi_0}$.
\end{lemma}
\begin{proof}
We have:
\begin{align}
	\vec{w} A U & = \lambda_w \vec{w} = \lambda_w \vec{w} U,
\end{align}
via \lemmaRef{vUv}, so:
\begin{align}
	(\vec{w} A - \lambda_w \vec{w}) U & = 0.
\end{align}
Now, this implies via \lemmaRef{vU0}:
\begin{align}
	(\vec{w} A - \lambda_w \vec{w}) & = d \vec{\psi_0},
	\label{eq:conditionOnEigenvector}
\end{align}
for some scalar $d$.

If $d=0$, $\vec{w} A = \lambda_w \vec{w}$ and $\lambda_w$ is an eigenvalue of $A$ with eigenvector $\vec{w} \neq \vec{\psi_0}$ (since $\vec{w} \perp \vec{\psi_0}$), as required.

If $d \neq 0$, then $\vec{v} = \frac{\lambda_w - \lambda_0}{d} \vec{w} + \vec{\psi_0}$ is an eigenvector of $A$ with eigenvalue $\lambda_w$ (proof by substitution). 
Here, if $\lambda_w \neq \lambda_0$ then $\vec{v} \neq \vec{\psi_0}$, as required.
We can then show that if $\lambda_w = \lambda_0$ (in which case both equal 1, and we would have $\vec{v} = \vec{\psi_0}$) cannot occur since it leads to a contradiction, as follows.
\eq{conditionOnEigenvector} with $\lambda_w = 1$ gives:
\begin{align}
	\vec{w} A - \vec{w} & = d \psi_0, \\
(\times A:)\ \ 	\vec{w} A^2 - \vec{w} A & = d \psi_0,
\end{align}
and summing these two equations gives:
\begin{align}
\vec{w} A^2 - \vec{w} & = 2 d \psi_0.
\end{align}
With further iterations we have:
\begin{align}
\vec{w} A^n - \vec{w} & = n d \psi_0,
\end{align}
for integers $n \geq 1$.
Now, taking the limit as $n \rightarrow \infty$, the RHS clearly diverges for any $d \neq 0$, whereas the LHS does not (via Gelfand's thereom \cite{horn13}, the norm of $A^n$ scales as $\rho(A)^n=1$ as $n \rightarrow \infty$).
This contradiction implies that we cannot have $d \neq 0$ with $\lambda_w = \lambda_0 = 1$ here.

As such, we have shown that $\lambda_w$ is also an eigenvalue of $A$ for some eigenvector $\vec{v} \neq \vec{\psi_0}$ here.
\end{proof}

\section{Centering projected covariance matrix in presence of zero-mode eigenvalue}
\label{app:convergenceUTOmegaU}

Here we demonstrate how to write a convergent form for $\Omega_U = U^T \Omega U$ when the zero-mode $\vec{\psi_0}$ is an eigenvector of $C$ with eigenvalue $\lambda_0=1$.

\subsection{Continuous-time case}
\label{app:convergenceUTOmegaUContinuous}

Modifying the derivation by Barnett et al. \cite{barn09b} for $\Omega$ for the continuous-time process in \eq{ornsteinUhlenbeck}, we first approximate \eq{ornsteinUhlenbeck} by the discrete-time process:
\begin{align}
	\vec{x}(t + dt) = \vec{x}(t) [I - \theta (I - C) dt] + \zeta \vec{r}(t) \sqrt{dt}
	\label{eq:discreteApproxToContinuous},
\end{align}
where $\vec{r}(t)$ is uncorrelated mean-zero unit-variance Gaussian noise.
We then right multiply by $U$, along with the substitution $K=I-\theta (I-C)dt$:
\begin{align}
	\vec{x}(t + dt) U = \vec{x}(t) K U + \zeta \vec{r}(t) U \sqrt{dt}
	\label{eq:discreteApproxToContinuousProjected},
\end{align}
and then, since $K$ also has $\vec{\psi_0}$ as an eigenvector with eigenvalue 1, we use \eq{CUequalsUCU} to restate:
\begin{align}
	\vec{x}(t + dt) U = \vec{x}(t) U K U + \zeta \vec{r}(t) U \sqrt{dt}
	\label{eq:covarianceProjectedContinuousUnaveraged},
\end{align}

We left multiply \eq{covarianceProjectedContinuousUnaveraged} by its transpose, and average over the ensemble at a given time $t$ to obtain:
\begin{align}
	U^T \overline{\vec{x}^T(t + dt) \vec{x}(t + dt)} U = & U^T K^T U^T \overline{\vec{x}^T(t) \vec{x}(t)} U K U \nonumber \\ & + \zeta^2 U^T U dt
	\label{eq:UOmegaUEnsembleAverage},
\end{align}
since $\overline{\vec{r}^T(t)\vec{r}(t)} = I$ for all $t$, and all cross-terms vanish because $\vec{r}(t)$ is uncorrelated with $\vec{x}(t)$.

Extending \cite{barn09b}, we require stationarity of $\vec{x}(t)U$ such that $\Omega_U = U^T \overline{\vec{x}^T(t) \vec{x}(t)} U = U^T \overline{\vec{x}^T(t + 1) \vec{x}(t + 1)} U$.
Stationarity of $\vec{x}(t)U$ in \eq{discreteApproxToContinuousProjected} for fixed $dt$ requires $\rho(KU) < 1$.
This is satisfied where we have $|\lambda_{K}| < 1$ for all eigenvalues $\lambda_{K}$ of $K$ except that corresponding to $\vec{\psi_0}$, via \theoremRef{specRadiusReducedBelow1} in \app{unaveragingZeroMode}.
Now, the eigenvalues of $K$ are $\lambda_K = 1 - \theta (1 - \lambda_C)dt$ with $\lambda_C$ the eigenvalues of $C$. As such, we need $| 1 - \theta (1 - \lambda_C)dt | < 1$ which in the continuous limit gives the stationarity condition $\mathrm{Re}(\lambda_C) < 1$, for all eigenvalues $\lambda_C$ of $C$ except that corresponding to $\vec{\psi_0}$.

Then we have the following expression (from \eq{UOmegaUEnsembleAverage}):
\begin{align}
	\Omega_U = U K^T \Omega_U K U + \zeta^2 U dt
	\label{eq:bracketedCovarianceRecurrenceExpressionContinuous}.
\end{align}
We then substitute $K=I-\theta(I-C)dt$ back in:
\begin{align}
	\Omega_U = &\ U (I-\theta(I-C^T)dt) \Omega_U (I-\theta(I-C)dt) U + \zeta^2 U dt, \nonumber \\
			= &\ U \Omega_U U - U (I-C^T) \Omega_U U \theta dt - U \Omega_U (I - C) U \theta dt \nonumber \\
			  &\ + \zeta^2 U dt + O(dt^2), \nonumber \\
			= &\ \Omega_U - U (I-C^T) \Omega_U \theta dt - \Omega_U (I - C) U \theta dt \nonumber \\
			  &\ + \zeta^2 U dt + O(dt^2), \nonumber \\
		0  = &\ - \Omega_U dt + U C^T \Omega_U dt - \Omega_U dt + \Omega_U C U dt \nonumber \\
		      &\ + \frac{\zeta^2}{\theta} U dt + O(dt^2), \nonumber 
\end{align}
and consider terms at highest order $O(dt)$ to find that in the continuous limit $dt \rightarrow 0$:
\begin{align}
	2 \Omega_U = &\ \frac{\zeta^2}{\theta} U + U C^T \Omega_U + \Omega_U C U
	\label{eq:bracketedCovarianceRecurrenceExpressionContinuous2}, \\
			= &\ \frac{\zeta^2}{\theta} U + (C U)^T \Omega_U + \Omega_U C U
	\label{eq:bracketedCovarianceRecurrenceExpressionContinuous3}.
\end{align}

In the special case when $C$ is symmetric (and with $CU$ having no eigenvalue equal to 1), we can solve for $\Omega_U$ to obtain:
\begin{align}
	\Omega_U = \frac{\zeta^2}{2\theta} \frac{U}{I - CU}
	\label{eq:bracketedCovarianceSymmetricContinuous}.
\end{align}
Otherwise, we obtain the following power series solution:
\begin{align}
	2 \Omega_U = & \ \frac{\zeta^2}{\theta} U + \frac{\zeta^2}{2\theta} ((CU)^T U + U (C U)) + \nonumber \\
	 & \frac{\zeta^2}{4\theta} \left[ ((CU)^T)^2 U + 2(CU)^T U (CU) + U (CU)^2 \right] + \ldots
	, \nonumber \\
	 = & \frac{\zeta^2}{\theta} \sum_{m=0}^{\infty}{ 2^{-m} \sum_{u=0}^{m}{\binom{m}{u} ((CU)^u)^T U (CU)^{m-u}} }
	\label{eq:bracketedCovarianceGeneralContinuous},
\end{align}
insofar as it converges.
Then, we simplify via \eqs{CmUEqualsCUm}{UCmTEqualsCUmT}, and then \eq{CUequalsUCU}:
\begin{align}
	2 \Omega_U = & \frac{\zeta^2}{\theta} \sum_{m=0}^{\infty}{ 2^{-m} \sum_{u=0}^{m}{\binom{m}{u} U (C^u)^T U C^{m-u}} U}
	\label{eq:bracketedCovarianceGeneralContinuous2}, \\
			= & \frac{\zeta^2}{\theta} \sum_{m=0}^{\infty}{ 2^{-m} \sum_{u=0}^{m}{\binom{m}{u} U (C^u)^T C^{m-u}} U}
	\label{eq:bracketedCovarianceGeneralContinuousFinal}.
\end{align}

Finally, as per \cite{barn09b}, the above stationarity condition does not guarantee convergence of \eqs{bracketedCovarianceGeneralContinuous}{bracketedCovarianceGeneralContinuousFinal}. 
As such, we briefly demonstrate that a sufficient condition for convergence is $| \lambda_C | < 1$, for all eigenvalues $\lambda_C$ of $C$ except that corresponding to $\vec{\psi_0}$ (which implies the stationarity condition).
For any matrix norm $|| \cdot ||$ \cite{horn13}
applied to \eq{bracketedCovarianceGeneralContinuousFinal}, we have:
\begin{align}
	2 || \Omega_U || & \leq \frac{\zeta^2}{\theta} \sum_{m=0}^{\infty}{ 2^{-m} \sum_{u=0}^{m}{\binom{m}{u} || ((C U)^u)^T || \ || (C U)^{m-u} || } }
	\label{eq:bracketedCovarianceGeneralContinuousNorms},
\end{align}
using \eqs{CmUEqualsCUm}{UCmTEqualsCUmT}, that the product of norms is greater than the norm of the products, and that the sum of norms is greater than the norm of the sum.
Similarly, we have:
\begin{align}
	2 || \Omega_U || & \leq \frac{\zeta^2}{\theta} \sum_{m=0}^{\infty}{ 2^{-m} \sum_{u=0}^{m}{\binom{m}{u} || (C U)^T ||^u \ || C U ||^{m-u} } }
	\label{eq:bracketedCovarianceGeneralContinuousNorms2}.
\end{align}
It is well known that  for all $\epsilon > 0$ there exists a matrix norm $|| \cdot ||$ such that $|| A || \leq \rho(A) + \epsilon$ \cite[Lemma 5.6.10]{horn13}.
Noting $\rho( (CU)^T ) = \rho( CU )$, we observe that for any $\epsilon > 0$ there exists a matrix norm such that:
\begin{align}
	2 || \Omega_U || \leq \frac{\zeta^2}{\theta} \sum_{m=0}^{\infty}{ ( \rho(C U) + \epsilon)^{m} }
	\label{eq:bracketedCovarianceGeneralContinuousNorms3}.
\end{align}
Then, if we have $\rho(C U) < 1$ and choose $\epsilon$ such that $(\rho(C U) + \epsilon)$ remains $< 1$, then \eq{bracketedCovarianceGeneralContinuousNorms3} converges, leaving $|| \Omega_U ||$ finite.
As such, convergence of this sum of norms then implies convergence of the matrix sum for $\Omega_U$ \cite{horn13},
under the condition $\rho(C U) < 1$.
This is satisfied where $|\lambda_{C}| < 1$ for all eigenvalues $\lambda_{C}$ of $C$ except that corresponding to $\vec{\psi_0}$, via \theoremRef{specRadiusReducedBelow1} in \app{unaveragingZeroMode}.

\subsection{Discrete-time case}
\label{app:convergenceUTOmegaUDiscrete}

Using a parallel analysis for the discrete process \eq{discreteARProcess}, we right multiply \eq{discreteARProcess} by $U$:
\begin{align}
	\vec{x}(t + 1) U = \vec{x}(t) C U + \zeta \vec{r}(t) U
	\label{eq:covarianceProjectedDiscrete},
\end{align}
and use \eq{CUequalsUCU} to restate:
\begin{align}
	\vec{x}(t + 1) U = \vec{x}(t) U C U + \zeta \vec{r}(t) U
	\label{eq:covarianceProjectedDiscreteUnaveraged}.
\end{align}

We then left multiply \eq{covarianceProjectedDiscreteUnaveraged} by its transpose, and average over the ensemble at a given time $t$ to obtain:
\begin{align}
	U^T \overline{\vec{x}^T(t + 1) \vec{x}(t + 1)} U = U ^T C^T U^T \overline{\vec{x}^T(t) \vec{x}(t)} U C U + \zeta^2 U^T U
	\label{eq:UOmegaUEnsembleDiscreteAverage},
\end{align}
since $\overline{\vec{r}^T(t)\vec{r}(t)} = I$ for all $t$, and all cross-terms vanish because $\vec{r}(t)$ is uncorrelated with $\vec{x}(t)$.
As above, we require stationarity of $\vec{x}(t) U$ such that $\Omega_U = U^T \overline{\vec{x}^T(t) \vec{x}(t)} U = U^T \overline{\vec{x}^T(t + 1) \vec{x}(t + 1)} U$.
For stationarity of $\vec{x}(t) U$ in \eq{covarianceProjectedDiscreteUnaveraged} we need $\rho(C U) < 1$, which is met when we have $|\lambda_{C}| < 1$ for all eigenvalues of $C$ except that corresponding to $\vec{\psi_0}$ (via \theoremRef{specRadiusReducedBelow1} in \app{unaveragingZeroMode}).

Then we have the following expression for $\Omega_U = U^T \Omega U$:
\begin{align}
	\Omega_U = &\ U C^T \Omega_U C U + \zeta^2 U
	\label{eq:bracketedCovarianceRecurrenceExpressionDiscrete}, \\
			= &\ (C U)^T \Omega_U C U + \zeta^2 U
	\label{eq:bracketedCovarianceRecurrenceExpressionDiscrete2}.	
\end{align}
In the special case when $C$ is symmetric (and since we have $\rho(C U) < 1$ from the stationarity condition above), we can solve for $\Omega_U$ to obtain:
\begin{align}
	\Omega_U =  \frac{\zeta^2 U}{I - (CU)^2}
	\label{eq:bracketedCovarianceSymmetricDiscrete}.
\end{align}
Otherwise, we then obtain the following power series solution:
\begin{align}
	\Omega_U = &\ \zeta^2 \left[ U + (C U)^T U (C U) + ((C U)^2)^T U (C U)^2 + \ldots \right]
	, \nonumber \\
	 = &\ \zeta^2 \sum_{u=0}^{\infty}{((C U)^u)^T U (C U)^{u}}
	\label{eq:bracketedCovarianceGeneralDiscrete},
\end{align}
insofar as it converges.
We again simplify via \eqs{CmUEqualsCUm}{UCmTEqualsCUmT}, and \eq{CUequalsUCU} to obtain:
\begin{align}
	\Omega_U = &\ \zeta^2 \sum_{u=0}^{\infty}{U (C^u)^T C^{u} U}
	\label{eq:bracketedCovarianceGeneralDiscreteFinal}.
\end{align}

Unlike the continuous case, the stationarity condition here does indeed guarantee convergence of \eqs{bracketedCovarianceGeneralDiscrete}
{bracketedCovarianceGeneralDiscreteFinal}.
This is demonstrated via a similar argument with matrix norms as was used for the continuous case.

\section{$\left\langle \sigma^2 \right\rangle$ for discrete time}
\label{app:discreteExpansion}

Expanding \eq{synchronizabilityDiscrete} in a similar way to the continuous process in \fullSubSecRef{motifRelationships}{synchronizabilityMotifsContinuous} via \eq{traceInOrthogonalBookend} we obtain:
\begin{align}
	\left\langle \sigma^2 \right\rangle & = \frac{\zeta^2}{N} \sum_{u=0}^{\infty}{\left( \sum_{i}{ ((C^u)^T C^{u})_{ii}} - \frac{1}{N}\sum_{i,j}{ ((C^u)^T C^{u})_{ij}} \right)}, \nonumber
\end{align}
so:
\begin{align}
	\left\langle \sigma^2 \right\rangle & = \frac{\zeta^2}{N} \sum_{u=0}^{\infty}{\left( \sum_{i,k}{ C^u_{ki} C^{u}_{ki}} - \frac{1}{N}\sum_{i,j,k}{ C^u_{ki} C^{u}_{kj}} \right)}, \nonumber \\
	& = \frac{\zeta^2}{N} \sum_{u=0}^{\infty}{\left( \sum_{i,k}{ \walkprod{k}{i}{u}{k}{i}{u}	} - \frac{1}{N}\sum_{i,j,k}{ \walkprod{k}{i}{u}{k}{j}{u}	} \right)}, \nonumber \\
	& = \zeta^2(1-\frac{1}{N} ) + \frac{\zeta^2}{N} \sum_{u=1}^{\infty}{ \sum_{i,k}{ \left( \walkprod{k}{i}{u}{k}{i}{u} - \frac{1}{N} \sum_{j}{ \walkprod{k}{i}{u}{k}{j}{u} }  \right) } }
	\label{eq:synchronizabilityMotifsDiscreteAppendix}.
\end{align}
using \eq{weightedDualWalkMotifCounts}.

\section{Proof of convergent forms for symmetric $C$}
\label{app:formForSymmetricC}

Starting from \eq{synchronizabilityContinuous} for continuous time, with $\theta=\zeta=1$ and symmetric $C = C^T$ we have:
\begin{align}
	\left\langle \sigma^2 \right\rangle & = \sum_{m=0}^{\infty} \frac{2^{-m-1}}{N} \sum_{u=0}^{m} \binom{m}{u} \mathrm{trace}\left( U C^m U \right),
	\nonumber \\
	& = \sum_{m=0}^{\infty} \frac{1}{2N} \mathrm{trace}\left( U C^m U \right), \nonumber \\
	& = \frac{1}{2N} \sum_{m=0}^{\infty} \mathrm{trace}\left( (C U)^m \right), \nonumber
\end{align}
via \eq{CUequalsUCU} and \eq{CmUEqualsCUm}.
Then, since:
\begin{itemize}
\item the trace of a matrix is equal to the sum of its eigenvalues; 
\item the eigenspectrum $\lambda_{CU}$ of $C U$ is the same as that of $C$, $\lambda_C$, except $\lambda_0$ for $\psi_0$ if $\psi_0$ is an eigenvector, which will have $\lambda_0 \rightarrow 0$ for $CU$ (see \lemmaRef{zeroModeAfterProjection} and \lemmaRef{orthogonalEigenvectorsAfterProjection} in \app{unaveragingZeroMode}); and 
\item the eigenvalues of $C^m$ are those of $C$ raised to the power $m$,
\end{itemize}
we have:
\begin{align}
	\left\langle \sigma^2 \right\rangle & = \frac{1}{2N} \sum_{m=0}^{\infty} \sum_{\lambda \neq \lambda_0} { \lambda_C^m }, \nonumber \\
	& =  \frac{1}{2N} \sum_{\lambda \neq \lambda_0} \sum_{m=0}^{\infty} { \lambda_C^m }, \nonumber \\
	& =  \frac{1}{2N} \sum_{\lambda \neq \lambda_0} \frac{1}{ 1 - \lambda_C },
	\label{eq:symmetricCcontinuous} 
\end{align}
with the last step valid since $|\lambda_C| < 1$ for all $\lambda_C$ corresponding to eigenvectors other than $\psi_0$ (being the domain of validity of our solution in \eq{bracketedCovarianceGeneralContinuousMain}).

Indeed, one can see a simpler route to this same result via \eq{bracketedCovarianceSymmetricContinuous}, however we have chosen to start from \eq{synchronizabilityContinuous} to demonstrate explicitly how the power series converges.

Similarly, starting from \eq{synchronizabilityDiscrete} for discrete time, with $\zeta=1$ and symmetric $C = C^T$ we have:
\begin{align}
	\left\langle \sigma^2 \right\rangle & = \frac{1}{N} \sum_{u=0}^{\infty}{ \mathrm{trace}\left( U C^{2u} U \right)}.
	\nonumber \\
	& = \frac{1}{N} \sum_{u=0}^{\infty} \mathrm{trace}\left( (C U)^{2u} \right), \nonumber
\end{align}
via \eq{CUequalsUCU} and \eq{CmUEqualsCUm}.
Then, following similar reasoning as for the continuous-time case, we have:
\begin{align}
	\left\langle \sigma^2 \right\rangle & = \frac{1}{N} \sum_{u=0}^{\infty} \sum_{\lambda \neq \lambda_0} { \lambda_C^{2u} }, \nonumber \\
	& =  \frac{1}{N} \sum_{\lambda \neq \lambda_0} \sum_{u=0}^{\infty} { (\lambda_C^{2})^{u} }, \nonumber \\
	& =  \frac{1}{N} \sum_{\lambda \neq \lambda_0} \frac{1}{ 1 - \lambda_C^2 },
	\label{eq:symmetricCdiscrete} 
\end{align}
with the last step valid since $|\lambda_C| < 1$ for all $\lambda_C$ corresponding to eigenvectors other than $\psi_0$ (being the domain of validity of our solution in \eq{bracketedCovarianceGeneralContinuousMain}).

\section{Supplementary figures}
\label{app:suppFigs}

\begin{figure}
		\includegraphics[width=0.95\columnwidth]{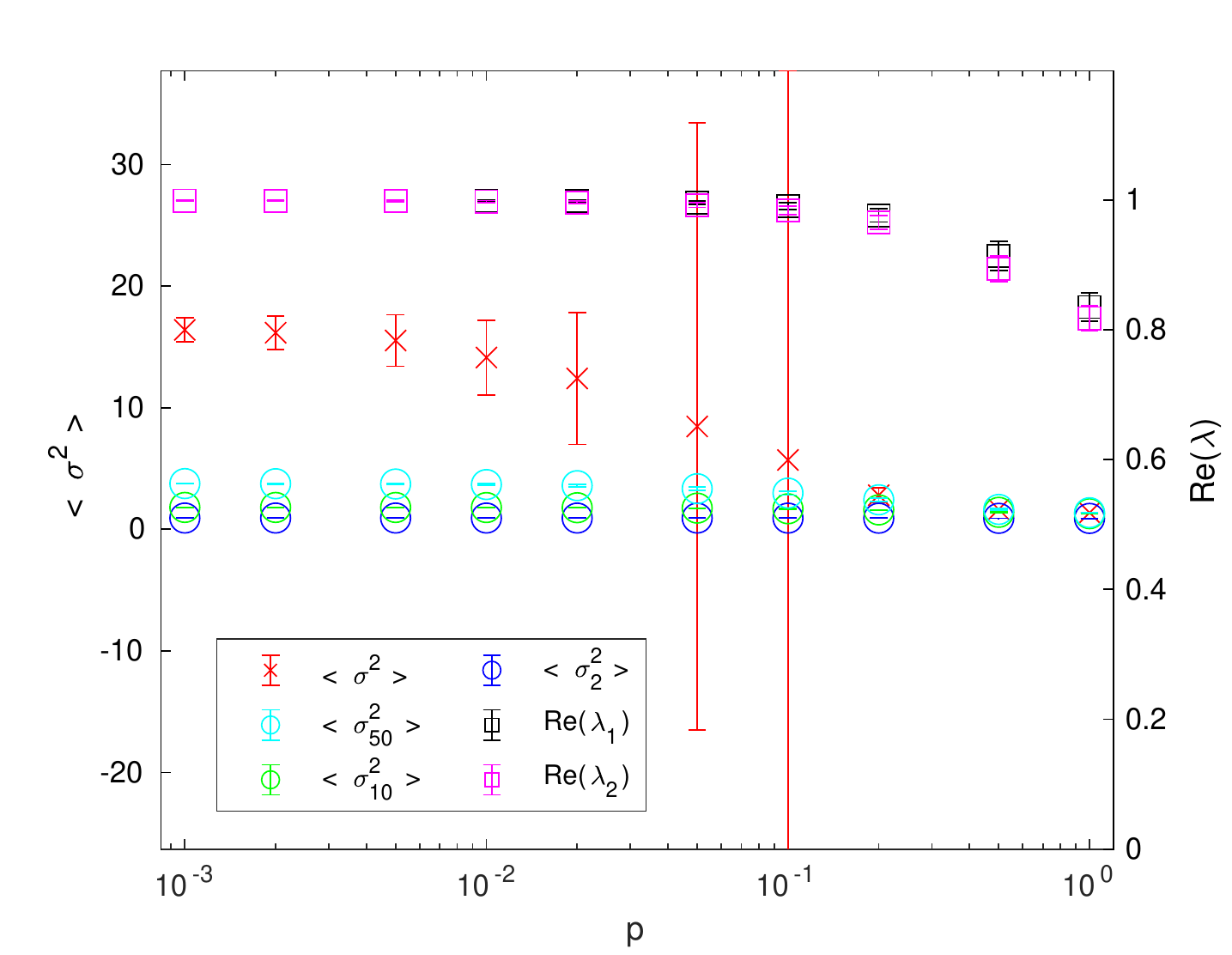}
	\caption{\label{fig:d2_c0_5_fullErrorBars} \textit{Numerical results for $d=2$, $c=0.5$} following \fig{d2_c0_5}, with the y-axis expanded to show the full extent of the error bars.
}
\end{figure}

\fig{d2_c0_5_fullErrorBars} reproduces \fig{d2_c0_5}, but with the y-axis expanded to show the full extent of the error bars.
The error bars are large in this regime because with $d=2$ a small amount of randomization $p$ has a relatively higher probability (as compared to larger $d$ and larger $p$) of making the network almost disconnected, and therefore driving $\left\langle \sigma^2 \right\rangle$ significantly higher for some network samples.

\end{document}